\documentclass[lettersize,journal]{IEEEtran}
\newif\ifshowbios
\showbiostrue

\makeatletter
\newcommand{\authorbio}[3]{%
  \par\addvspace{14pt}%
  \noindent\begin{minipage}[t]{\columnwidth}%
    \def\@IEEEBIOskipN{0pt}%
    \IfFileExists{#1}%
      {\begin{IEEEbiography}[{\includegraphics[width=1in,height=1.25in,clip,keepaspectratio]{#1}}]{#2}#3\end{IEEEbiography}}%
      {\begin{IEEEbiography}{#2}#3\end{IEEEbiography}}%
  \end{minipage}\par}
\makeatother
\newif\ifshortitemref
\DeclareRobustCommand{\enumref}[2]{\ifshortitemref(#2)\else#1(#2)\fi}
\newcommand{\shortref}[1]{\begingroup\shortitemreftrue\ref{#1}\endgroup}

\usepackage[T1]{fontenc}
\usepackage{xcolor}
\usepackage{comment}
\usepackage{bm}
\usepackage{mdframed}
\usepackage{acronym}
\usepackage{graphicx}
\usepackage{epstopdf}
\usepackage{amsmath}
\usepackage{algorithm}

\usepackage{algpseudocode}
\usepackage{amssymb}
\usepackage{amsfonts}
\usepackage{mathtools}
\usepackage{physics} 
\usepackage{float}
\usepackage{booktabs}
\usepackage{url}
\usepackage{enumitem}     
\usepackage[hidelinks]{hyperref}
\usepackage{cleveref}
\usepackage{tikz}
\usepackage[caption=false,font=footnotesize]{subfig}

\usepackage{longtable}
\usepackage{orcidlink}
\usepackage{threeparttable}
\usepackage{tablefootnote}
\usepackage{microtype}
\usepackage{cite}
\usetikzlibrary{shapes, shapes.symbols, arrows.meta, positioning, calc, fit, backgrounds, shadows}
\usetikzlibrary{patterns}

\graphicspath{{figures/}}

\newtheorem{theorem}{Theorem}
\newtheorem{proposition}{Proposition}
\newtheorem{lemma}{Lemma}
\newtheorem{corollary}{Corollary}
\newtheorem{definition}{Definition}

\newtheorem{remark}{Remark}

\crefname{theorem}{Theorem}{Theorems}
\Crefname{theorem}{Theorem}{Theorems}
\crefname{proposition}{Proposition}{Propositions}
\Crefname{proposition}{Proposition}{Propositions}
\crefname{lemma}{Lemma}{Lemmas}
\Crefname{lemma}{Lemma}{Lemmas}
\crefname{corollary}{Corollary}{Corollaries}
\Crefname{corollary}{Corollary}{Corollaries}
\crefname{definition}{Definition}{Definitions}
\Crefname{definition}{Definition}{Definitions}
\crefname{assumption}{Assumption}{Assumptions}
\Crefname{assumption}{Assumption}{Assumptions}
\crefname{remark}{Remark}{Remarks}
\Crefname{remark}{Remark}{Remarks}
\crefname{figure}{Fig.}{Figs.}
\Crefname{figure}{Fig.}{Figs.}

\newcommand{\ORCID}[1]{\,\orcidlink{#1}}
\newcommand{\ignore}[1]{}

\makeatletter
\newcommand*\rel@kern[1]{\kern#1\dimexpr\macc@kerna}
\newcommand*\widebar[1]{%
  \begingroup
  \def\mathaccent##1##2{%
    \rel@kern{0.8}%
    \overline{\rel@kern{-0.8}\macc@nucleus\rel@kern{0.2}}%
    \rel@kern{-0.2}%
  }%
  \macc@depth\@ne
  \let\math@bgroup\@empty \let\math@egroup\macc@set@skewchar
  \mathsurround\z@ \frozen@everymath{\mathgroup\macc@group\relax}%
  \macc@set@skewchar\relax
  \let\mathaccentV\macc@nested@a
  \macc@nested@a\relax111{#1}%
  \endgroup
}

\makeatother

\begin{document}
\title{Pricing for Information Revelation in Demand Response:\\ A Strategic Communication Approach}

\author{
    Hassan~Munif\ORCID{0009-0009-8838-243X},
    Chao~Zhang\ORCID{0000-0002-1469-9911},
    Samson~Lasaulce\ORCID{0000-0001-9837-9538}, \IEEEmembership{Member,~IEEE,}
    Olivier~Beaude\ORCID{0000-0003-1719-8749},
    Vineeth~S.~Varma\ORCID{0000-0001-8762-2790},
    Mounir~Ghogho\ORCID{0000-0002-0055-7867}, \IEEEmembership{Fellow,~IEEE,}
    and~H.~Vincent~Poor\ORCID{0000-0002-2062-131X}, \IEEEmembership{Life~Fellow,~IEEE}
    \thanks{
        % (\emph{Corresponding author: Chao Zhang.}) \\
        Hassan Munif, Samson Lasaulce, and Vineeth S. Varma are with the CRAN Laboratory, Universit\'e de Lorraine, CNRS, F-54000 Nancy, France.}
    \thanks{Chao Zhang is with Central South University, Changsha, China.}
    \thanks{Olivier Beaude is with EDF Lab Paris-Saclay, OSIRIS Department, Palaiseau, France.}
    \thanks{Mounir Ghogho is with the College of Computing, Mohammed VI Polytechnic University, Rabat, Morocco.}
    \thanks{H. Vincent Poor is with the Department of Electrical and Computer Engineering,
        Princeton University, Princeton, NJ 08544 USA.}
}

\markboth{IEEE Transactions on Smart Grid}
% \markboth{Journal of \LaTeX\ Class Files,~Vol.~xx, No.~xx, DEC. 2026}
{Munif \MakeLowercase{\textit{et al.}}: Pricing for Information Revelation in Demand Response: A Strategic Communication Approach}
\maketitle

\begin{abstract}
    Demand response aggregators require consumers' private flexibility information to dispatch loads efficiently.
    Yet self-interested consumers have incentives to misreport it. This paper investigates how much information heterogeneous consumers can credibly convey through nonbinding, unverifiable reports, and how a pre-announced retail tariff can improve this revelation. We model this interaction as a multi-sender strategic communication game with state-dependent bias, using the cheap-talk framework. The aggregator maximizes social welfare, and consumers have private valuations and quadratic payoffs. Closed-form conditions identify when the set of dispatched consumers is invariant to the reports. Under these conditions, reporting incentives decouple into independent single-sender subgames. We derive a closed-form uniform tariff that simultaneously minimizes the mean-squared strategic bias of all dispatched consumers. In this regime, the welfare loss relative to full communication is expressed as a weighted sum of residual posterior variances, showing that every equilibrium weakly improves expected social welfare over no communication. Numerical studies support the theoretical results and show that, when consumers coordinate on informative equilibria, strategic communication with the proposed bias-minimizing tariff recovers most of the full-communication welfare gain over no communication.
\end{abstract}


\section*{Nomenclature}

\subsection*{Sets and Indices}
\begin{IEEEdescription}[\IEEEusemathlabelsep\IEEEsetlabelwidth{$a_n^{\mathrm{eff}}, b_n^{\mathrm{eff}}$}]
    \item[$\mathcal{V}, N$] Consumer index set; number of consumers.
    \item[$\mathcal{A}, \mathcal{I}$] Sets of active and inactive consumers.
    \item[$\Omega_n$] Type support $[\underline{\omega}_n, \overline{\omega}_n]$ of Consumer~$n$.
\end{IEEEdescription}

\subsection*{Parameters}
\begin{IEEEdescription}[\IEEEusemathlabelsep\IEEEsetlabelwidth{$a_n^{\mathrm{eff}}, b_n^{\mathrm{eff}}$}]
    \item[$a, b$] Coefficients of the operating cost $C(\cdot)$.
    \item[$a_n^{\mathrm{eff}}, b_n^{\mathrm{eff}}$] Effective cost coefficients of Consumer~$n$.
    \item[$\alpha_n$] Satiation parameter.
    \item[$\omega_n, \psi_n$] Private type (valuation); its prior density.
\end{IEEEdescription}

\subsection*{Variables and Functions}
\begin{IEEEdescription}[\IEEEusemathlabelsep\IEEEsetlabelwidth{$a_n^{\mathrm{eff}}, b_n^{\mathrm{eff}}$}]
    \item[$x_n, X$] Consumer dispatch; aggregate consumption.
    \item[$u_{\mathrm{S},n}, u_{\mathrm{R}}$] Consumer net payoff; Aggregator payoff.
    \item[$p, p^*$] Retail tariff; bias-minimizing tariff.
    \item[$m_n, \sigma_n$] Message; signaling strategy.
    \item[$g_n,\gamma_n,\delta_n$] Receiver preference map; slope and intercept.
    \item[$\widehat{\omega}_n$] Posterior mean $\mathbb{E}[\omega_n \mid m_n]$.
    \item[$\theta_n, \phi_n$] Canonical type; its density.
    \item[$ \omega_n^\star, \theta_n^\star$] Agreement type in original and canonical spaces.
    \item[$\kappa_n, \!\kappa_{n,\mathrm{max}}$] Equilibrium message count; feasible supremum.
    \item[$\mu_{n,i}$] Partition boundaries in the canonical type space.
\end{IEEEdescription}
\subsection*{\texorpdfstring{Abbreviations}{Abbreviations}}
\begin{IEEEdescription}[\IEEEusemathlabelsep\IEEEsetlabelwidth{$a_n^{\mathrm{eff}}, b_n^{\mathrm{eff}}$}]
    \item[$\text{FC, NC}$] Full/no communication.
    \item[$\mathrm{RW}$] Recovered-welfare metric.
    \item[{SIT}] Strategic information transmission.
    \item[{P-BRD}] Projected best-response dynamics.
\end{IEEEdescription}

\section{Introduction}
\label{sec:introduction}
\IEEEPARstart{T}{he} transformation of modern power systems through renewable energy integration and transportation electrification creates a coordination challenge for grid operators \cite{zhu2025revisiting, beaude2016reducing}. Real-time balancing requires flexible modulation of consumption across a large set of privately owned distributed energy resources (such as thermostatic loads, battery storage, and electric vehicles), each characterized by heterogeneous preferences and coordinated by \emph{aggregators} that deliver collective grid services. To optimize system-wide outcomes, aggregators must know each consumer's operational flexibility, including capacities for load curtailment, shifting, or rescheduling. However, consumers are self-interested strategic agents whose private valuations are not directly observable. This creates an information asymmetry: while the aggregator requires local flexibility data to optimize dispatch, consumers have strategic incentives to misreport their preferences for financial or operational advantage. In modern smart grids, this reporting can be executed automatically by home energy management systems or software agents optimizing under this misalignment.

This tension between \emph{information asymmetry} and \emph{strategic misalignment} arises in many demand response (DR) programs \cite{parag2016electricity, kiedanski2020misalignments}; the core mechanism-design question is therefore: \emph{how to elicit credible information transmission from strategic participants in a practically implementable way?} Existing DR mechanisms trade off implementation complexity, incentive properties, and autonomy/privacy preservation \cite{sankar2013smart}.
Direct load control programs grant utilities direct curtailment authority, but at the cost of consumer autonomy, limiting adoption in residential settings where control over consumption is highly valued \cite{saad2016toward}. Incentive-compatible mechanisms such as Vickrey-Clarke-Groves mechanisms \cite{samadi2012advanced}, self-reported baseline mechanisms \cite{muthirayan2020mechanisma}, or contract-theoretic designs \cite{parvizi2025contract-based} incentivize truthful reporting through well-designed monetary transfers, with recent variants addressing privacy and computational concerns \cite{satchidanandan2023atwo-stage,tsaousoglou2020truthful}; however, their implementation can entail additional communication, computation, and settlement requirements. Other approaches, such as billing mechanisms \cite{jacquot2019analysis}, real-time pricing \cite{ding2019distributed}, peer-to-peer trading \cite{guo2024ahigh-efficiency}, and energy sharing \cite{chen2021approaching, chen2020an, chen2021an, yang2025robust}, coordinate consumption through prices, bids, and trading rules.

In this paper, we investigate how much information consumers can credibly convey through nonbinding flexibility reports for centralized dispatch. For example, the operator publishes a day-ahead tariff and a short menu of urgency levels; each household selects a level, and the operator dispatches the corresponding load adjustment in real time. Were consumer and aggregator objectives aligned, this exchange would be a classical encoding-decoding problem; under misalignment, credible information transmission is itself an equilibrium property. In game theory, this is a \emph{cheap talk} game in the sense of Crawford and Sobel \cite{crawford1982strategic}: messages carry no inherent cost or commitment, yet can convey information at equilibrium.
Strategic communication models have seen growing interest in energy systems. The cheap-talk framework was applied to smart-grid DR in a single-sender setting~\cite{larrousse2014crawford}, while related work explores mechanism design under asymmetric information~\cite{rayati2022electricity}. Another line uses Bayesian persuasion for DR~\cite{aid2025signalling}, building on committed information design~\cite{kamenica2011bayesian}.

This paper develops a tractable \emph{multi-sender}, \emph{state-dependent-bias} cheap-talk formulation for \emph{heterogeneous} DR populations, using the retail tariff for information design. Multi-consumer settings pose three challenges. First, \emph{multi-sender strategic coupling}: with convex costs, one consumer's allocation depends on aggregate demand \cite{samadi2012advanced}, so each consumer's signaling depends on beliefs about the others' strategies, creating a generally intractable high-dimensional fixed-point problem \cite{mcgee2013cheap,ambrus2008multi-sender,velicheti2025value}. Second, \emph{state-dependent bias}: the widely adopted quadratic utility framework \cite{samadi2012advanced} yields misreporting incentives that vary with consumer urgency, unlike classical cheap-talk models with constant bias \cite{crawford1982strategic,saritas2017quadratic}. Third, \emph{heterogeneity}: real populations exhibit diverse valuation distributions. These challenges raise the following questions: First, when does the coupled
multi-sender game admit a tractable, consumer-by-consumer
analysis? Second, how efficient can strategic communication be when consumers have different preferences? And third, how does pricing affect signaling strategies of heterogeneous consumers?
We address these questions through the following contributions:

1) \emph{Tractable Analysis for Strategic DR.} We prove that, under closed-form parameter conditions, the multi-sender signaling game admits a tractable decomposition into independent single-sender subgames, despite the strategic coupling induced by aggregate DR decisions. The result also covers mixed populations by identifying consumers that are inactive.\\
2) \emph{Tariff as an Information-Design Lever.} We show that a pre-announced retail tariff shapes consumers' strategic reporting bias and determines the availability of informative equilibria. We then derive a closed-form uniform tariff that minimizes the system-wide mean-squared strategic bias, to improve information revelation under strategic reporting.\\
3) \emph{Equilibrium Characterization and Welfare Guarantees.} For log-concave priors, we establish uniqueness of each feasible finite-message partition equilibrium and global convergence of a projected best-response algorithm. We express the welfare loss relative to full communication as a weighted sum of residual posterior variances, showing that every equilibrium weakly improves expected social welfare over no communication. Numerical studies assess welfare recovery under specified equilibrium selection and compare tariff-design error sensitivity with a price-taking and a naive-trust benchmark.

The paper proceeds as follows. \Cref{sec:n_consumer_game} formulates the DR signaling game; \Cref{sec:decoupling} establishes decoupling; \Cref{sec:single_consumer_analysis} characterizes equilibria, tariff design, and welfare; \Cref{sec:simulations} evaluates the framework numerically; and \Cref{sec:conclusion} concludes.

\noindent\textbf{Notation.} Scalars use plain font ($x_n, p$), vectors bold font ($\bm{x}, \bm{m}$), matrices bold uppercase ($\bm{Q}$), sets calligraphic font ($\mathcal{V}$), and spaces blackboard bold ($\mathbb{M}$). The subscript $-n$ denotes exclusion of $n$, applied to sets
($\mathcal{V}_{-n},\mathcal{A}_{-n}$) or profiles ($\bm{\omega}_{-n}$).

% \clearpage

\section{Strategic Communication \texorpdfstring{Model}{Model}}
\label{sec:n_consumer_game}

\subsection{System Model}
\label{subsec:system_model}
We analyze a single time slot with $N$ consumers indexed by $\mathcal{V}\coloneqq\{1,\dots,N\}$ and one aggregator. Each consumer's smart meter enables two-way information exchange with the Aggregator. We model payoff functions under the widely adopted quadratic framework used in DR (see, e.g., \cite{guo2024ahigh-efficiency,chen2021approaching,ding2019distributed,samadi2012advanced,samadi2010optimal}). At the considered DR event, Consumer $n$ has a private parameter $\omega_n$ modeling its consumption valuation for the DR event; we refer to $\omega_n$ as Consumer $n$'s \emph{type}. The valuation is drawn from a continuous, strictly positive prior density $\psi_n$ with bounded support $\Omega_n \coloneqq [\underline{\omega}_n, \overline{\omega}_n] \subset \mathbb{R}$, independently across consumers. The Aggregator knows the priors $\{\psi_n,\Omega_n\}_{n\in\mathcal{V}}$, but not the realized valuations $({\omega_n})_{n\in\mathcal{V}}$.

\emph{a) Consumer preferences and utility.}
Consumer $n$ derives utility from consuming $x_n$ units via a concave function with linearly decreasing marginal benefit \cite{samadi2012advanced}:
\begin{equation}
    u_{\mathrm{C},n}(\omega_n, x_n) = \omega_n x_n - \frac{\alpha_n}{2} x_n^2,
    \label{eq:consumer_intrinsic_benefit}
\end{equation}
where $\alpha_n>0$ is a fixed \emph{satiation} parameter. In our formulation, $x_n$ denotes the controllable consumption for the DR event.
% \footnote{The baseline-adjustment interpretation is detailed in \Cref{subsec:model_scope}.}
The dispatched $x_n$ (decided by the Aggregator) is assumed to remain below the satiation point, $x_n \le \omega_n/\alpha_n$, beyond which marginal utility is negative (see, e.g., \cite{guo2024ahigh-efficiency}).
Given the tariff $p$, Consumer $n$'s net payoff is
\begin{equation}
    u_{\mathrm{S},n}(\omega_n, x_n)\coloneqq u_{\mathrm{C},n}(\omega_n, x_n)-p\,x_n.
    \label{eq:consumer_net_utility}
\end{equation}

\emph{b) Aggregator's objective.}
Let $\bm{x} \coloneqq (x_1, \dots, x_N)$, $\bm{\omega} \coloneqq (\omega_1, \dots, \omega_N)$, and total load $X \coloneqq \sum_{n\in\mathcal{V}} x_n$. The Aggregator is a regulated, welfare-maximizing entity; it maximizes social welfare, the aggregate intrinsic benefit minus the operating cost (see, e.g., \cite{samadi2012advanced}):
\begin{align}
    u_{\mathrm{R}}(\bm{\omega}, \bm{x}) & \coloneqq \sum_{n\in\mathcal{V}} u_{\mathrm{C},n}(\omega_n, x_n)-C(X), \label{eq:Aggregator_welfare_utility} \\
    \text{where } C(X)                  & \coloneqq aX^2+bX, \quad a>0,\;\; b\geq 0. \label{eq:Aggregator_cost}
\end{align}
The cost $C(X)$ captures procurement/operating costs with increasing marginal cost \cite{samadi2012advanced} and $C'(X)= 2aX + b$ is the marginal operating cost of the grid.

\begin{remark}[Regulated aggregator and internal transfers]
    \label{rem:price_internal_transfer}
    The Aggregator is a regulated, welfare-maximizing entity; retail payments are internal transfers: they change incentives through \eqref{eq:consumer_net_utility} but cancel from total surplus in \eqref{eq:Aggregator_welfare_utility}, leaving the welfare objective independent of $p$. Thus, the tariff $p$ plays an information-design role as it can be used to improve information transmission and control misreporting in regulated utilities and cooperative, community, or DSO-operated programs. For a profit-seeking aggregator as in \cite{alshehri2021agame-theoretic}, the objective would include revenue. If utilities remain quadratic, the strategic communication equilibrium analysis could still apply via the canonical-form reduction of \Cref{prop:game_equivalence}; more general utility functions are a possible extension of this work.
\end{remark}

\begin{figure}[!t]
    \centering
    \resizebox{0.98\linewidth}{!}{
        \begin{tikzpicture}[
                font=\sffamily,
                >=Latex,
                node distance=0.45cm and 2.4cm,
                % --- Styles ---
                myBlue/.style={draw=blue!40!black, fill=blue!5, text=black},
                myOrange/.style={draw=orange!50!black, fill=orange!10, text=black},
                sourceNode/.style = {rectangle, draw=gray!60, fill=white, thick, rounded corners, minimum width=1.0cm, align=center},
                edgeLabel/.style = {fill=white, inner sep=1.5pt, font=\large, text opacity=1, rounded corners=2pt},
                stepNum/.style = {circle, draw=black, fill=white, inner sep=0.5pt, font=\normalsize\bfseries, text=black}
            ]

            % --- Nodes ---
            % Central Consumer n -- (Sender)
            \node[rectangle, rounded corners, thick, align=center, minimum height=0.6cm, myOrange, minimum width=2.8cm, drop shadow] (cons_n) {\textbf{\large Consumer} $\bm{n}$\\ \normalsize(Sender)};

            % Dummy node for spacing
            \node[below=0.1cm of cons_n] (knowledge) {};

            % Other Consumers (faded)
            \node[rectangle, rounded corners, draw=orange!30, fill=orange!2, text=orange!30!black, align=center, minimum height=0.5cm, above=0.4cm of cons_n, minimum width=2.5cm] (cons_1) {Consumer $1$};
            \node[rectangle, rounded corners, draw=orange!30, fill=orange!2, text=orange!30!black, align=center, minimum height=0.5cm, below=0.4cm of cons_n, minimum width=2.5cm] (cons_N) {Consumer $N$};

            % Dots
            \path (cons_1) -- (cons_n) node[pos=0.3, font=\large, gray] {$\vdots$};
            \path (cons_n) -- (cons_N) node[pos=0.3, font=\large, gray] {$\vdots$};

            % ---  Prior Distribution Source ---
            \node[sourceNode, left=2.5cm of cons_n] (nat) {\textbf{\large Prior}\\$\psi_n$};

            % Aggregator
            \node[rectangle, rounded corners, thick, align=center, minimum height=0.3cm, myBlue, right=2.9cm of cons_n, minimum height=2.0cm, minimum width=2.0cm, drop shadow] (agg) {\textbf{\large Aggregator}\\ \normalsize(Receiver)};

            % Population Box
            \begin{pgfonlayer}{background}
                \node[fit=(cons_1)(cons_N)(knowledge), draw=orange!50, dashed, thick, rounded corners=10pt, fill=orange!5!white, inner sep=10pt] (pop_box) {};
            \end{pgfonlayer}
            % \node[anchor=north, font=\normalsize\bfseries, text=orange!40!black] at (pop_box.south) {Heterogeneous Population};

            % --- ARROWS ---

            % 1. Price (From Aggregator to Population) - The "Design" phase
            \draw[->, dashed, thick, blue!60!black] (agg.north)
            .. controls +(up:1.3cm) and +(up:1.3cm) .. (pop_box.north)
                node[midway, above=0.14cm, edgeLabel, text=blue!60!black] {\textcircled{\normalsize 1} Sets Ex Ante Tariff $p$};

            % 2. Stochastic Draw (To Consumers) -- 
            \draw[->, thick, gray!80] (nat.east) -- (cons_n.west)
            node[pos=0.58, above, edgeLabel, text=black] {$\omega_n \sim \psi_n$};

            \draw[->, thick, gray!80] (nat.east) -- (cons_n.west)
            node[midway, below=0.03, edgeLabel, text=black] {\textcircled{\normalsize 2}};

            % Faded arrows for others
            \draw[->, thick, gray!30] (nat.east) -- (cons_1.west);
            \draw[->, thick, gray!30] (nat.east) -- (cons_N.west);

            % 3. Message (Consumer to Aggregator)
            \draw[->, very thick, red!70!black] ([yshift=3mm]cons_n.east) -- ([yshift=3mm]agg.west)
                node[midway, above=0.1, edgeLabel, text=red!70!black] {\textcircled{\normalsize 3} Message $m_n$};

            % 4. Allocation (Aggregator to Consumer)
            \draw[<-, very thick, green!40!black] ([yshift=-3mm]cons_n.east) -- ([yshift=-3mm]agg.west)
                node[midway, below=0.1, edgeLabel, text=green!40!black] {\textcircled{\normalsize 4} Dispatch $x_n$};

        \end{tikzpicture}
    }
    \caption{Framework overview: (1) the Aggregator sets an ex ante tariff $p$; (2) types $\omega_n \sim \psi_n$ are realized; (3) consumers strategically send messages $m_n$; (4) the Aggregator dispatches $x_n$.
    }
    \label{fig:framework_overview}
\end{figure}

\subsection{\texorpdfstring{Modeling Assumptions and Scope}{Modeling Assumptions and Scope}}
\label{subsec:model_scope}
To isolate the strategic-communication problem, we have adopted several simplifying assumptions. We discuss their rationale and the scope for extensions below.

(1) \emph{Aggregate Cost and Utility Functions:} We adopt quadratic consumer utilities and a quadratic aggregate operating cost, following standard DR models \cite{samadi2012advanced,guo2024ahigh-efficiency}. The aggregate cost $C(X)$ couples consumers through total load, while the quadratic structure yields an affine dispatch in posterior means, which enables the decoupling in \Cref{sec:decoupling}. More general smooth concave utilities and smooth convex costs admit local quadratic approximations, although extending the results to these settings requires a separate analysis.
% In the presence of non-convex operating effects, such as unit commitment and start-up costs, the formulation serves as a convex economic-dispatch approximation.
With non-convex effects such as unit commitment and start-up costs, the formulation provides a convex economic-dispatch approximation.

(2) \emph{The Private Type:} The private type $\omega_n$ captures Consumer $n$'s event-specific consumption valuation. After the tariff is announced, the consumer observes this valuation and sends a message about its urgency or willingness to consume. The parameter $\alpha_n$ is known and captures heterogeneity in utility curvature and satiation. If $\alpha_n$ were also private, a scalar $\theta_n\coloneqq(\omega_n-p)/\alpha_n$ would still capture the consumer's ideal action, but the Aggregator's ideal action would depend separately on $\alpha_n$. This would require extending the analysis to signaling with multidimensional private types~\cite{saritas2017quadratic,battagliniMultipleReferralsMultidimensional2002}.

(3) \emph{Single-Slot Formulation:}\label{rem:repeated} We study strategic communication within a single DR event, such as an hour-ahead dispatch window, following common single-slot DR models \cite{muthirayan2020mechanisma,tushar2014prioritizing}. When types are independently redrawn, and payoffs and constraints are separable across time, a stage-game PBE can be repeated independently across slots. Accounting for intertemporal coupling through storage or deferrable loads (e.g., electric-vehicle charging) would require a dynamic signaling model; we leave this extension for future work.

(4) \emph{Baseline Form:}\label{rem:baseline} The formulation also covers controllable adjustments around a fixed DR baseline. For example, let total consumption be $q_n={q}^0_n+x_n$, where ${q}^0_n$ is the pre-event baseline and $x_n$ is the controllable adjustment. Denote the valuation and linear cost coefficients in the total-consumption model by $\omega_n'$ and $b'$, respectively. Up to action-independent terms, substituting for $q_n$ recovers the quadratic form in \eqref{eq:consumer_intrinsic_benefit}--\eqref{eq:Aggregator_cost} under the relabeling $\omega_n=\omega_n'-\alpha_n{q}^0_n$ and $b=b'+2a\sum_j{q}^0_j$. The subsequent results apply when the adjustments are represented by $x_n\geq 0$ and the transformed model satisfies the stated parameter and feasibility assumptions.

(5) \emph{Network Constraints:}\label{rem:net_cont} We consider a localized residential area (or a small microgrid), whose aggregate load represents a small share of the total demand in a distribution network. In this setting,
we assume that voltage and line-flow constraints are nonbinding (see, e.g.,  \cite{chen2021approaching,wang2016cooperative}). Binding network constraints would require an analysis with a nodal dispatch formulation.

\subsection{Game \texorpdfstring{Timeline and Equilibrium Concept}{Timeline and Equilibrium Concept}}
\label{subsec:game_formulation}
The DR communication stage is a one-shot multi-sender cheap-talk game (\Cref{fig:framework_overview}):
\begin{enumerate}[wide]
    \item \emph{Tariff setting (ex ante):} The Aggregator (a regulated utility) announces a fixed uniform retail tariff $p$ for the upcoming slot (e.g., a day-ahead tariff).
          % $p$ becomes common knowledge.

    \item \emph{Type realization:} Nature draws each private type $\omega_n \sim \psi_n$ independently.

    \item \emph{Strategic signaling:} Consumer $n$, observing only its own private type $\omega_n$, sends a message $m_n\in \mathbb{M}$ to the Aggregator, where $\mathbb{M}$ is a sufficiently rich message space.\footnote{In practice, $\mathbb{M}$ can be a finite menu of urgency or flexibility levels (e.g., `High Flexibility', `Moderate', or `Inflexible'). In the partition equilibria of \Cref{sec:single_consumer_analysis}, a message identifies an interval of $\omega_n$ rather than its exact value, which limits information disclosure.}

    \item \emph{Dispatch decision:} The Aggregator observes $\bm{m}=(m_1,\dots,m_N)$ (not $\bm{\omega}$) and chooses the dispatch $\bm{x}\in\mathbb{R}^N_{\geq 0}$.
\end{enumerate}

\begin{remark}[Ex Ante Role of the Tariff]
    \label{rem:price_design_lever}
    The tariff $p$ acts as an information-design lever: rather than a market-clearing price, it is set in advance to influence reporting incentives and the informativeness of supportable equilibria.
    It shapes information through payoffs, rather than by committing to a disclosure policy as in Bayesian persuasion \cite{kamenica2011bayesian}.
    This separation of timescales (fixing the tariff before type realization and dispatch) aligns with the operational hierarchy of power markets, where regulated tariffs are determined ahead of real-time balancing.
\end{remark}

To formalize this interaction, we restrict attention to deterministic strategies. \footnote{We focus on deterministic strategies, standard in cheap-talk models (e.g., \cite{saritas2017quadratic}): with a continuous prior, randomization either concerns measure-zero (boundary) types or mixes among messages inducing the same action; in both cases, it changes neither the induced type-to-action map nor expected welfare.}
A sender strategy for Consumer $n$ is a mapping $\sigma_n: \Omega_n \to \mathbb{M}$ from types to messages. The Aggregator's receiver strategy is a dispatch rule $\bm{x}(\bm{m}): \mathbb{M}^N \to \mathbb{R}^N_{\geq 0}$ acting on the message profile $\bm{m} \coloneqq (m_1, \dots, m_N)$. Let $\pi(\cdot \mid \bm{m})$ denote the Aggregator's posterior belief over the type profile $\bm{\omega}$. Derived from this belief system, the Aggregator forms a \emph{posterior-mean mapping} $\widehat{\omega}_n: \mathbb{M} \to \Omega_n$ for each consumer, which assigns to any received message $m_n$ the conditional expectation $\widehat{\omega}_n(m_n) \coloneqq \mathbb{E}[\omega_n \mid m_n]$. The solution concept is a Perfect Bayesian Equilibrium (PBE).

\begin{definition}[Perfect Bayesian Equilibrium]
    \label{def:pbe}
    A strategy profile $\bigl((\sigma_n^\ast)_{n \in \mathcal{V}},\, \bm{x}^\ast(\cdot)\bigr)$ and a posterior belief system $\pi^\ast$ constitute a PBE if they satisfy:

    1) \noindent\emph{Sequential rationality (consumers):} given $\omega_n$ and others' equilibrium policies, Consumer $n$ maximizes expected payoff:
    \begin{equation}
        \sigma_n^\ast(\omega_n) \in \arg\max_{m_n \in \mathbb{M}} \mathbb{E}_{\bm{\omega}_{-n}}\left[ u_{\mathrm{S},n} \big(\omega_n, x_n^\ast(m_n, \bm{m}_{-n}) \big) \right], \label{eq:pbe_sender}
    \end{equation}
    where the expectation is over the types of the other consumers $\bm{\omega}_{-n} \coloneqq (\omega_j)_{j \neq n}$ (drawn from priors $\psi_j$), and $\bm{m}_{-n} \coloneqq (\sigma_j^\ast(\omega_j))_{j \neq n}$ denotes the corresponding messages.

    2) \noindent \emph{Sequential rationality (Aggregator):} for any observed $\bm{m}$, the Aggregator maximizes expected welfare under its posterior:
    \begin{equation}
        \bm{x}^\ast(\bm{m}) \in \arg\max_{\bm{x} \in \mathbb{R}_{\ge 0}^N} \mathbb{E}_{\bm{\omega} \sim \pi^\ast(\cdot \mid \bm{m})}\left[u_{\mathrm{R}}(\bm{\omega}, \bm{x})\right]. \label{eq:pbe_receiver}
    \end{equation}

    3) \noindent \emph{Belief consistency:} $\pi^\ast(\bm{\omega} \mid \bm{m})$ follows from the priors and the strategies $(\sigma_n^\ast)_{n \in \mathcal{V}}$ by Bayes' rule whenever possible.
\end{definition}
Although messages are costless and unverifiable, equilibrium ensures their credibility: any unilateral deviation from $\sigma_n^*$ is unprofitable given the others' strategies.

The PBE receiver problem in \eqref{eq:pbe_receiver} has a useful structure:
because the Aggregator's objective is quadratic and its uncertainty enters
linearly through the types, the optimal dispatch depends on messages only
through posterior means.
We now characterize the Aggregator's dispatch rule \eqref{eq:pbe_receiver} and identify the operating regimes in which the multi-consumer cheap-talk game decouples into independent single-sender problems.

\section{Dispatch Participation and Strategic Decoupling}
\label{sec:decoupling}
Because the Aggregator's operating cost depends on aggregate load, the dispatch each consumer receives depends on the reports of others, which couples their reporting incentives. We show that this coupling collapses whenever the set of consumers receiving a positive dispatch does not change with the reports. We derive closed-form parameter conditions that guarantee this invariance, allowing each consumer's reporting problem to be analyzed independently. We begin by characterizing the Aggregator's dispatch rule.

\subsection{Aggregator's Dispatch Rule and Strategic Coupling}
\label{subsec:equilibrium_interaction}
Since the Aggregator's utility function $u_{\mathrm{R}}$ is quadratic in $\bm{x}$ and linear in $\bm{\omega}$, its optimal dispatch rule depends on the posterior only through the vector $\widehat{\bm{\omega}} \coloneqq (\widehat{\omega}_1, \dots, \widehat{\omega}_N)$, where $\widehat{\omega}_n = \mathbb{E}[\omega_n \mid \bm{m}]$.
Relaxing the non-negativity constraints, the first-order condition of the expected social welfare $\mathbb{E}[u_\mathrm{R} \mid \bm{m}]$ with respect to $x_n$ yields $\widehat{\omega}_n - \alpha_n x_n - C'(X) = 0$. Thus, Consumer $n$'s \emph{unconstrained dispatch} is
\begin{equation}
    x_n^{\mathrm{unc}}(\widehat{\bm{\omega}}) = \frac{\widehat{\omega}_n - C'(X^*)}{\alpha_n},
    \label{eq:individual_dispatch_marginal}
\end{equation}
where $X^*$ is the aggregate unconstrained optimal load. A sufficient condition for the non-negativity constraint not to bind is $\widehat{\omega}_n>C'(X^*)$ for any admissible posterior $\widehat{\omega}_n$.

Solving the system of first-order conditions \eqref{eq:individual_dispatch_marginal} over all consumers yields the following closed form for Consumer $n$'s \emph{unconstrained dispatch} (see \Cref{app:Aggregator_br_derivation} for full derivation):
\begin{equation}
    x_{n,\mathcal{V}}^{\mathrm{unc}}(\widehat{\bm{\omega}})
    =
    \frac{\widehat{\omega}_n - b}{\alpha_n}
    - \underbrace{\frac{2a}{\alpha_n}
        \frac{\sum_{j\in \mathcal{V}} (\widehat{\omega}_j - b)/\alpha_j}
        {1 + 2a \sum_{j\in\mathcal{V}} (1/\alpha_j)}}_{\text{strategic coupling}}.
    \label{eq:Aggregator_best_response}
\end{equation}

The first term is consumer-specific. The underbraced term is a coupling term: through the aggregate marginal cost, Consumer $n$'s dispatch depends on the Aggregator's posterior beliefs about all consumers. Consequently, each consumer's optimal message depends on their beliefs about the other consumers' strategies, resulting in a high-dimensional fixed-point problem.

Because \eqref{eq:Aggregator_best_response} is affine in the posterior means, we show that this strategic interaction ultimately simplifies and the game decouples. However, this decoupling relies on the dispatch remaining in the affine regime.
If a non-negativity constraint binds, the dispatch rule becomes piecewise affine, and strategic coupling can reappear. This motivates characterizing when the active set of consumers induced by the non-negativity constraints remains fixed over all admissible posterior means.

\subsection{\texorpdfstring{Dispatch Participation Regimes}{Dispatch Participation Regimes}}
\label{subsec:active_sets_regimes}

Because the optimal dispatch vector $\bm{x}^*(\widehat{\bm{\omega}})$ must satisfy the physical non-negativity constraints, the Aggregator's dispatch decisions implicitly split the set of consumers. For any posterior-mean vector $\widehat{\bm{\omega}}$ in the joint belief space $\bm{\Omega} \coloneqq \Omega_1 \times \cdots \times \Omega_N$, we define the disjoint sets of active and inactive consumers as:
\begin{align}
     & \text{\emph{Active consumer set}:}~\mathcal{A}(\widehat{\bm{\omega}})  \coloneqq \{n \in \mathcal{V} : x_n^*(\widehat{\bm{\omega}}) > 0\},   \\
     & \text{\emph{Inactive consumer set}:}~\mathcal{I}(\widehat{\bm{\omega}})  \coloneqq \{n \in \mathcal{V} : x_n^*(\widehat{\bm{\omega}}) = 0\},
\end{align}
where $\mathcal{V} = \mathcal{A}(\widehat{\bm{\omega}}) \cup \mathcal{I}(\widehat{\bm{\omega}})$. Here, physical dispatch participation refers strictly to receiving a non-zero allocation ($x_n^* > 0$); all consumers, active or inactive, still participate in the strategic communication stage by submitting messages.

In general, the identities of the active and inactive consumers can change dynamically depending on the realized belief vector $\widehat{\bm{\omega}}$. However, we are particularly interested in cases where this split remains invariant. We say the system has a \emph{fixed active set} $\mathcal{A}$ if $\mathcal{A}(\widehat{\bm{\omega}}) = \mathcal{A}$ (and consequently $\mathcal{I}(\widehat{\bm{\omega}}) = \mathcal{I}$) for every admissible posterior-mean vector $\widehat{\bm{\omega}} \in \bm{\Omega}$. This allows us to identify the following dispatch participation regimes:
\begin{definition}[Dispatch Participation Regimes]
    \label{def:participation_regimes}
    The DR program operates in one of three dispatch participation regimes:
    \begin{enumerate}[label=(\roman*),wide,ref=\enumref{\thedefinition}{\roman*}]
        \item \label[definition]{def:participation_regimes:i} \emph{Full participation}: The active set is fixed and includes the entire population, i.e., $\mathcal{A}(\widehat{\bm{\omega}}) = \mathcal{V}$ and $\mathcal{I}(\widehat{\bm{\omega}}) = \emptyset,$ $\forall\widehat{\bm{\omega}} \in \bm{\Omega}$.

        \item \label[definition]{def:participation_regimes:ii} \emph{Partial participation}: The active set is fixed but includes a subset of the population, i.e., $\emptyset \neq \mathcal{A}(\widehat{\bm{\omega}}) = \mathcal{A} \subsetneq \mathcal{V},$ $\forall\widehat{\bm{\omega}} \in \bm{\Omega}$.

        \item \label[definition]{def:participation_regimes:iii} \emph{Varying participation}: The active set $\mathcal{A}(\widehat{\bm{\omega}})$ is dynamic and varies with the realized posterior belief vector $\widehat{\bm{\omega}} \in \bm{\Omega}$.
    \end{enumerate}
\end{definition}

While varying participation introduces nonlinearities that couple consumers' strategies, full and partial participation preserve a structure that allows tractable analysis. This motivates characterizing the conditions where a fixed split is guaranteed.

\subsection{\texorpdfstring{Certifying a Fixed Active Set}{Certifying a Fixed Active Set}}
\label{subsec:fixed_active_conditions}
Verifying whether a candidate split $(\mathcal{A}, \mathcal{I})$ remains fixed over the continuous, high-dimensional belief space $\bm{\Omega}$ is difficult. However, because the unconstrained dispatch is strictly monotonic in the posterior beliefs, we can establish exact analytical conditions by analyzing only the worst-case boundary vertices of the belief space.
% $X_{\mathcal{S}} \coloneqq \sum_{j\in\mathcal{S}} x_j$,
% $\overline{\omega}_i \le 2a X_{\mathcal{A}}^{\min} + b=C'(X_{\mathcal{A}}^{\min})$, $\underline{\omega}_n > 2a X_{\mathcal{A}_{-n}}^{\max} + b=C'(X_{\mathcal{A}_{-n}})$.
The following lemma formalizes this stability requirement.

\begin{lemma}[Fixed Active-Set Condition]
    \label{lemma:certificate}
    A candidate split $(\mathcal{A},\mathcal{I})$ with $\mathcal{A}\neq\emptyset$ is a fixed active/inactive split for every admissible posterior $\widehat{\bm{\omega}} \in \bm{\Omega}$ if and only if:
    % the system parameters satisfy
    \begin{align}
        \underline{\omega}_n
         & >
        b+2a \frac{\sum_{j \in \mathcal{A}_{-n}} (\overline{\omega}_j - b)/\alpha_j}{1 + 2a \sum_{j \in \mathcal{A}_{-n}} (1/\alpha_j)}, \quad  \forall n \in \mathcal{A},
        \tag{C1}
        \label{eq:two_tier_C1} \\[10pt]
        \overline{\omega}_i
         & \le
        b+2a \frac{\sum_{j\in\mathcal{A}} (\underline{\omega}_j - b)/\alpha_j}{1 + 2a \sum_{j\in\mathcal{A}} (1/\alpha_j)}, \quad \forall i\in\mathcal{I}.
        \tag{C2}
        \label{eq:two_tier_C2}
    \end{align}
\end{lemma}
\begin{IEEEproof}[Proof sketch]
    \IEEEQEDoff
    The unconstrained dispatch is affine and monotone in the posterior means. Hence, to certify that a candidate active set is fixed over the whole belief space, it is enough to check the boundary cases that minimize each active consumer's dispatch and maximize each inactive consumer's incentive to become active. These yield conditions \eqref{eq:two_tier_C1}--\eqref{eq:two_tier_C2}.
    See \Cref{app:lemma_certificate}.
\end{IEEEproof}

Conceptually, a split $(\mathcal{A}, \mathcal{I})$ remains fixed if every active consumer is worth serving even under the peak grid congestion generated by the rest of the active set, and inactive consumers remain not worth serving even under the minimal load:

(i) \emph{Condition \eqref{eq:two_tier_C1}:} The peak congestion generated by the other active consumers ($\mathcal{A}_{-n}$) occurs when they all report their maximum valuations $\overline{\omega}_j$. Active Consumer $n$'s lowest valuation $\underline{\omega}_n$ must still exceed this peak marginal cost.

(ii) \emph{Condition \eqref{eq:two_tier_C2}:} The minimal load generated by the active set ($\mathcal{A}$) occurs when they all report their minimum valuations $\underline{\omega}_j$. Inactive Consumer $i$'s highest valuation $\overline{\omega}_i$ must not exceed this minimal marginal cost.

Thus, full participation is certified when condition \eqref{eq:two_tier_C1} holds for all consumers with $\mathcal{A}=\mathcal{V}$ (and $\mathcal{I}=\emptyset$), which reduces to:

\begin{equation}
    \underline{\omega}_n > b+2a \frac{\sum_{j\in\mathcal{V}_{-n}}(\overline{\omega}_j-b)/\alpha_j}{1+2a\sum_{j\in\mathcal{V}_{-n}}(1/\alpha_j)}, \qquad \forall n\in\mathcal{V}.
    \tag{C0}
    \label{eq:certificate}
\end{equation}
Such boundary-based bounds are standard in DR mechanism design; for instance, the self-reported baseline mechanism of \cite{muthirayan2020mechanisma} similarly restricts its attention to active players.

% Since the optimal dispatch is unique for each $\widehat{\bm{\omega}}$, at most one split can be fixed; if it exists, it coincides with $\mathcal{A}(\widehat{\bm{\omega}})$ at any single admissible point, e.g., the prior means $\mathbb{E}[\bm{\omega}]\in\bm{\Omega}$. It therefore suffices to test the candidate $\mathcal{A}=\mathcal{A}(\mathbb{E}[\bm{\omega}])$: if it violates \eqref{eq:two_tier_C1}--\eqref{eq:two_tier_C2}, no fixed split exists.

\noindent The resulting decision flow for certifying dispatch participation regimes and strategic decoupling is summarized in \Cref{fig:regime_decision_flow}.

\begin{figure}[!t]
    \centering
    \resizebox{0.9\columnwidth}{!}{%
        \begin{tikzpicture}[
                >=Latex,
                font=\sffamily\footnotesize,
                node distance=0.52cm and 0.65cm,
                % Box Styles
                testBox/.style={
                        rectangle, draw=blue!50!black, fill=blue!3, rounded corners=2.5pt,
                        text width=3.3cm, align=center, inner sep=4.5pt, line width=0.6pt
                    },
                okBox/.style={
                        rectangle, draw=teal!60!black, fill=teal!4, rounded corners=2.5pt,
                        text width=4.4cm, align=left, inner sep=4.5pt, line width=0.6pt
                    },
                warnBox/.style={
                        rectangle, draw=orange!65!black, fill=orange!5, rounded corners=2.5pt,
                        text width=4.4cm, align=left, inner sep=4.5pt, line width=0.6pt
                    },
                tagYes/.style={
                        font=\bfseries\scriptsize, text=teal!60!black, above=1pt
                    },
                tagNo/.style={
                        font=\bfseries\scriptsize, text=red!70!black, left=2pt
                    }
            ]

            % --- Left Column (Conditions / Pipeline) ---
            \node[testBox] (t1) {
                \textbf{Check Condition \eqref{eq:certificate}}\\[1pt]
                \footnotesize All consumers active at peak congestion?
            };

            \node[testBox, below=0.52cm of t1] (t2) {
                \textbf{Check Conditions \eqref{eq:two_tier_C1}--\eqref{eq:two_tier_C2}}\\[1pt]
                % \scriptsize Fixed split $(\mathcal{A}, \mathcal{I})$ certified?
                \scriptsize Does some split $(\mathcal{A}, \mathcal{I})$ satisfy \eqref{eq:two_tier_C1}--\eqref{eq:two_tier_C2}?
            };

            % --- Right Column (Regimes) ---
            \node[okBox, right=0.65cm of t1] (r1) {
                \textbf{\textcolor{teal!45!black}{Full Participation}} $(\mathcal{A} = \mathcal{V})$
                % \\~\Cref{def:participation_regimes:i}
                \\[1.5pt]
                \scriptsize $\checkmark$ Exact strategic decoupling on all $\mathcal{V}$
            };

            \node[okBox, right=0.65cm of t2] (r2) {
                \textbf{\textcolor{teal!45!black}{Partial Participation}} $(\mathcal{A} \subsetneq \mathcal{V})$\\[1.5pt]
                \scriptsize $\checkmark$ Exact strategic decoupling on $\mathcal{A}$
            };

            \node[warnBox, below=0.52cm of r2] (r3) {
                \textbf{\textcolor{orange!55!black}{Varying Participation}}\\[1.5pt]
                \scriptsize \tikz[baseline=-0.5ex, line width=0.85pt, black!70!black]{\draw (-0.5ex,-0.5ex)--(0.5ex,0.5ex) (-0.5ex,0.5ex)--(0.5ex,-0.5ex);}\hspace{3pt}Decoupling guarantee does not apply; active set varies with beliefs
            };

            % --- Arrows ---
            % Test 1 branches
            \draw[->, thick, teal!60!black] (t1.east) -- (r1.west)
            node[midway, tagYes] {YES};

            \draw[->, thick, gray!65!black] (t1.south) -- (t2.north)
            node[midway, tagNo] {NO};

            % Test 2 branches
            \draw[->, thick, teal!60!black] (t2.east) -- (r2.west)
            node[midway, tagYes] {YES};

            % L-route into Regime 3
            \draw[->, thick, gray!65!black] (t2.south) |- (r3.west)
            node[pos=0.32, tagNo] {NO};

        \end{tikzpicture}%
    }
    \caption{Hierarchy of dispatch participation regimes and certification of strategic decoupling. The Aggregator evaluates conditions \eqref{eq:certificate} and \eqref{eq:two_tier_C1}--\eqref{eq:two_tier_C2} to identify whether the population admits a fixed active set, guaranteeing decoupling.}
    \label{fig:regime_decision_flow}
\end{figure}

\subsection{Strategic Decoupling on a Fixed Active Set}
\label{subsec:strategic_decoupling}
When the dispatch has a fixed active set $\mathcal{A}$, inactive consumers $i\in\mathcal{I}$ receive zero dispatch ($x_i^* = 0$) for all admissible posteriors, making their messages payoff-irrelevant. Consequently, the Aggregator's dispatch rule for the active consumers $n\in\mathcal{A}$ is equivalent to the unconstrained dispatch rule \eqref{eq:Aggregator_best_response} evaluated strictly over the active subpopulation $\mathcal{A}$:
\begin{equation}
    x_{n,\mathcal{A}}^*(\widehat{\bm{\omega}}) = x_{n,\mathcal{A}}^{\mathrm{unc}}(\widehat{\bm{\omega}}) \quad \forall n\in\mathcal{A}.
    \label{eq:active_dispatch_short}
\end{equation}

The dispatch rule \eqref{eq:active_dispatch_short} is affine in the posterior beliefs.
We now present the main structural result: in any dispatch with a fixed active set (i.e., full or partial participation), the multi-sender strategic coupling collapses, allowing each active consumer's game to be analyzed independently.

\begin{theorem}[Strategic Decoupling on a Fixed Active Set]
    \label{thm:strategic_independence}
    Suppose the dispatch has a fixed active set $\mathcal{A}$. For any fixed posterior-mean mapping $\widehat{\omega}_n(\cdot)$, the expected payoff ranking over messages for Consumer $n \in \mathcal{A}$ is independent of the signaling strategies of the other consumers $j \in \mathcal{A}_{-n}$.

    Consequently, the multi-sender game \emph{decouples}: Consumer $n$'s signaling problem becomes an independent single-sender game, parameterized by the other active consumers through their prior means $\mathbb{E}[\omega_j]$ and satiation parameters $\alpha_j$.
    % ${\{\mathbb{E}[\omega_j]\}_{j \in \mathcal{A}_{-n}}$.
\end{theorem}
\begin{IEEEproof}
    \IEEEQEDoff
    % With a fixed active set, the active dispatch is the affine rule \eqref{eq:active_dispatch_short}. When Consumer $n$ compares two of its own messages, the strategies of the other active consumers enter the expected payoff only through prior means and terms that do not depend on $n$'s chosen message. These message-independent terms cancel from all pairwise comparisons, so the entire ranking of Consumer $n$'s messages is independent of the other consumers' partitions.
    See \Cref{app:thm_two_tier}.
\end{IEEEproof}

This is a \emph{sender-side} decoupling. The realized Aggregator dispatch remains coupled across active consumers through the full message profile via \eqref{eq:active_dispatch_short}.
Thus, each active consumer's optimization problem decouples into an \emph{effective single-sender
    game} that can be solved separately. We define it below.

\begin{corollary}[Effective Single-Sender Game on a Fixed Active Set]
    \label{cor:equivalent_subgame}
    Assume the dispatch has a fixed active set $\mathcal{A}$. For every Consumer $n\in\mathcal{A}$, the game is strategically payoff-equivalent to an independent single-sender cheap-talk game with a sender utility \eqref{eq:consumer_net_utility} and a receiver effective utility:
    \begin{equation}
        u_{\mathrm{R},n}^{\mathrm{eff}}(\omega_n, x_n) \coloneqq \omega_n x_n - \left(\frac{\alpha_n}{2} + a_n^{\mathrm{eff}}\right)x_n^2 - b_n^{\mathrm{eff}}\,x_n,
        \label{eq:receiver_eff_util}
    \end{equation}
    where the effective parameters are constant for Consumer $n$:
    \begin{align}
        a_n^{\mathrm{eff}} & \coloneqq \frac{a}{1 + 2a \sum_{j \in \mathcal{A}_{-n}} (1/\alpha_j)},                                                                  \\[3pt]
        b_n^{\mathrm{eff}} & \coloneqq b + 2a \frac{\sum_{j \in \mathcal{A}_{-n}} (\mathbb{E}[\omega_j]-b)/\alpha_j}{1+2a\sum_{j \in \mathcal{A}_{-n}}(1/\alpha_j)}.
    \end{align}
\end{corollary}

\begin{IEEEproof}[Proof sketch]
    \IEEEQEDoff
    For Consumer $n$, the dispatch can be averaged over the other active consumers' posterior means. By the law of iterated expectations (LIE), these posterior means reduce to the prior means $\mathbb{E}[\omega_j]$, so their effect on Consumer $n$'s payoff is captured by the constants $a_n^{\mathrm{eff}}$ and $b_n^{\mathrm{eff}}$. The resulting sender payoff and receiver best response have the same quadratic single-sender form as \eqref{eq:receiver_eff_util}. See \Cref{app:proof_corollary}.
\end{IEEEproof}

Other consumers enter Consumer $n$'s problem only through the constants $a_n^{\mathrm{eff}}$ and $b_n^{\mathrm{eff}}$. Inactive consumers do not enter the effective subgame as their dispatch is zero.

\begin{remark}[When the Active Set Varies]
    \label{rem:marginal}
    If no fixed active set is certified by the conditions above, the non-negativity constraint binds for some posterior beliefs. Thus, decoupling can fail in this regime.
    \Cref{subsec:selection_robustness}
    quantifies the welfare discrepancy under the prescribed decoupled profile, which remains small in the tested cases.
\end{remark}

Having identified the participation regimes in which decoupling holds, we now analyze the resulting decoupled single-sender subgame.

\section{Analysis of the Single-Consumer Communication Subgame}
\label{sec:single_consumer_analysis}
The strategic independence result of \Cref{sec:decoupling} reduces the $N$-consumer game to parallel single-consumer subgames for the active consumers. In this section, we analyze one such game, namely the strategic interaction between Consumer $n$ (\emph{the Sender}) and the effective Aggregator (\emph{the Receiver}).

\subsection{Canonical Form for Quadratic Cheap Talk}
\label{subsec:canonical_form}

We begin with a generic quadratic cheap-talk game.\footnote{A quadratic specification can also be interpreted as a local second-order approximation of smooth utilities and costs around an operating point; the exact closed-form decoupling results rely on this quadratic structure.}
A single sender observes a type $\omega$, sends a costless message to a receiver, who then chooses an action $x$. The payoffs are quadratic and given by:
\begin{align}
    u_{\mathrm{S}}(\omega, x) & = s_2 x^2 + s_1(\omega) x + s_0(\omega), \label{eq:general_sender_utility}   \\
    u_{\mathrm{R}}(\omega, x) & = r_2 x^2 + r_1(\omega) x + r_0(\omega), \label{eq:general_receiver_utility}
\end{align}
where $s_2, r_2 < 0$ (strict concavity). The Sender and Receiver preferred actions are
$x_{\mathrm{S}}(\omega)
    \coloneqq
    -\tfrac{s_1(\omega)}{2s_2}$, and $ x_{\mathrm{R}}(\omega)
    \coloneqq
    -\tfrac{r_1(\omega)}{2r_2}$.
% \[
%     x_{\mathrm{S}}(\omega)
%     \coloneqq
%     -\frac{s_1(\omega)}{2s_2},
%     \qquad
%     x_{\mathrm{R}}(\omega)
%     \coloneqq
%     -\frac{r_1(\omega)}{2r_2}.
% \]

The strategic conflict is captured by the difference between these preferred actions, namely, the \emph{bias function}: $\beta(\omega) \coloneqq x_{\mathrm{S}}(\omega) - x_{\mathrm{R}}(\omega)$. The analysis of this game can be simplified by reparameterizing the type space as follows.

\begin{proposition}[Reduction to a Canonical Form]
    \label{prop:game_equivalence}
    Assume \(s_1(\omega)\) is strictly monotone. Then the game with payoffs \eqref{eq:general_sender_utility}--\eqref{eq:general_receiver_utility} is strategically equivalent to a canonical game with
    \[
        \widetilde u_{\mathrm{S}}(\theta,x)=-(x-\theta)^2,
        \qquad
        \widetilde u_{\mathrm{R}}(\theta,x)=-(x-g(\theta))^2,
    \]
    where the auxiliary type $\theta \coloneqq x_{\mathrm{S}}(\omega)$ is the Sender's preferred action, and the Receiver's ideal action is $g(\theta) \coloneqq x_{\mathrm{R}}\big(\omega(\theta)\big)$.
\end{proposition}
\begin{IEEEproof}
    The transformation $\theta= x_{\mathrm{S}}(\omega)$ is a bijection since $s_1(\omega)$ is monotonic. The Sender's utility $u_{\mathrm{S}}$ can be scaled by the positive constant $\frac{-1}{s_2}$ and shifted by a term independent of the action $x$ without altering the Sender's strategic behavior; this yields $\widetilde{u}_\mathrm{S}$. Similarly, scaling $u_{\mathrm{R}}$ by $\frac{-1}{r_2}$ yields $\widetilde{u}_\mathrm{R}$. As these are positive affine transformations of the utility functions, the equilibrium structure of the game is preserved.
\end{IEEEproof}
The canonical form normalizes the Sender's type to coincide with its preferred action (equivalently $\widetilde{u}_{\mathrm{S}}(\theta, x) = -(x - \theta)^2$); the bias becomes $\beta(\theta) \coloneqq \theta - g(\theta)$.

We now apply this single-sender, single-receiver strategic communication framework to Consumer $n$'s effective game from \Cref{cor:equivalent_subgame}. Consumer $n$'s and the Aggregator's utility functions are given by \eqref{eq:consumer_net_utility} and \eqref{eq:receiver_eff_util}, respectively. Let $\phi_n(\theta_n)$ denote the prior density induced by the type reparameterization. In this case, the canonical type $\theta_n$ and the receiver's response function $g_n$ are defined as follows:
\begin{equation}
    \theta_n = \frac{\omega_n-p}{\alpha_n}, \qquad
    g_n(\theta_n) = \frac{(\alpha_n\theta_n+p)-b_n^{\mathrm{eff}}}{\alpha_n+2a_n^{\mathrm{eff}}}.
\end{equation}
Writing $g_n(\theta_n) = \gamma_n\theta_n + \delta_n$, with
\begin{equation*}
    \gamma_n \coloneqq \frac{\alpha_n}{\alpha_n+2a_n^{\mathrm{eff}}}\quad \text{and} \quad \delta_n \coloneqq \frac{p-b_n^{\mathrm{eff}}}{\alpha_n+2a_n^{\mathrm{eff}}},
\end{equation*}
reveals that the bias $\beta_n(\theta_n)=\theta_n - g_n(\theta_n)$ is affine in $\theta_n$. Since $\alpha_n > 0$ and $a_n^{\mathrm{eff}} > 0$, we have $0 < \gamma_n < 1$. Consequently, unlike classical constant-bias models \cite{crawford1982strategic,saritas2017quadratic}, this framework yields a \emph{state-dependent} bias $\beta_n(\theta_n)=(1-\gamma_n)\theta_n-\delta_n$ (see, e.g., \cite{deimen2024communication,gordon2010on} for similar state-dependent bias studies).

\subsection{Equilibrium Structure and Characterization}
\label{subsec:eq_characterization}

To study and characterize the equilibrium, we first establish that the canonical model defined here satisfies the single-crossing \cite{athey2001single} property, which requires
\begin{equation*}
    \frac{\partial^2 \widetilde{u}_{\mathrm{S},n}}{\partial x_n\, \partial \theta_n} > 0.
\end{equation*}
In canonical coordinates, the Sender's utility is $-(x_n-\theta_n)^2$, yielding $\partial^2 \widetilde{u}_{\mathrm{S},n}/(\partial x_n\,\partial\theta_n) = 2 > 0$.
% Hence, for any two actions $x < x'$, the preference difference $\widetilde{u}_{\mathrm{S},n}(\theta_n, x') - \widetilde{u}_{\mathrm{S},n}(\theta_n, x)$ is strictly increasing in $\theta_n$.
Under this property, every finite equilibrium is outcome-equivalent to a partition of the Sender's type space into
convex intervals. This is a classic result for cheap-talk games satisfying the single-crossing property \cite{crawford1982strategic,gordon2010on}.

\begin{theorem}[Finite Partition Equilibria]
    \label{thm:general_partition_equilibrium}
    Consider Consumer \(n\)'s effective subgame. Every finite PBE is outcome-equivalent to an ordered interval partition of the type space $[\underline{\theta}_n, \overline{\theta}_n]$ into a finite set of $\kappa_n \ge 1$ convex intervals, defined by boundaries
    \[
        \underline{\theta}_n=\mu_{n,0}<\mu_{n,1}<\cdots
        <\mu_{n,\kappa_n}=\overline{\theta}_n ,
    \]
    % {that is, the vector of interior boundaries
    % $\bm{\mu}_n \coloneqq (\mu_{n,1},\dots,\mu_{n,\kappa_n-1})$ belongs to the
    % set of \emph{valid partitions}
    % \begin{equation}
    %     \mathcal{D}_n(\kappa_n)
    %     \coloneqq
    %     \left\{
    %     \bm{\mu}_n\in\mathbb{R}^{\kappa_n-1}
    %     :
    %     \underline{\theta}_n<\mu_{n,1}<\cdots<\mu_{n,\kappa_n-1}<\overline{\theta}_n
    %     \right\},
    %     \label{eq:valid_partition_set}
    % \end{equation}
    % with the convention $\mu_{n,0}=\underline{\theta}_n$ and
    % $\mu_{n,\kappa_n}=\overline{\theta}_n$ throughout.}

    Types in \(I_{n,i}\coloneqq[\mu_{n,i-1},\mu_{n,i}]\) send message
    \(m_{n,i}\), up to arbitrary tie-breaking at the boundaries, and the
    induced actions \(\{x_{n,i}^*\}_{i=1}^{\kappa_n}\) satisfy
    % \begin{align}
    %     &\text{Aggregator BR: }x_{n,i}^*
    %     =
    %     \mathbb{E}\!\left[g_n(\theta_n)\mid m_{n,i}\right],
    %     \quad i=1,\ldots,\kappa_n, \label{eq:best_response_theta}\\
    %     &\text{Consumer Indifference: }\mu_{n,i}
    %     =
    %     \frac{x_{n,i}^*+x_{n,i+1}^*}{2},
    %     \quad i=1,\ldots,\kappa_n-1 . \label{eq:arbitrage_cond_theta}
    % \end{align}

    \begin{align}
        x_{n,i}^* & = \mathbb{E}[g_n(\theta_n) \mid I_{n,i}], \quad \text{(Receiver action)} \label{eq:best_response_theta} \\
        \mu_{n,i} & = \frac{x_{n,i}^* + x_{n,i+1}^*}{2},\quad\text{ (Sender indifference)}  \label{eq:arbitrage_cond_theta}
    \end{align}
    for $i=1,\ldots,\kappa_n$ in \eqref{eq:best_response_theta} and $i=1,\ldots,\kappa_n-1$ in \eqref{eq:arbitrage_cond_theta}.
    Conversely, any strictly ordered partition satisfying
    \eqref{eq:best_response_theta}--\eqref{eq:arbitrage_cond_theta} supports a
    finite PBE of the effective subgame.

\end{theorem}

\begin{IEEEproof}[Proof sketch]
    \IEEEQEDoff
    In canonical coordinates, the Sender's payoff satisfies single-crossing because ${\partial^2 \widetilde{u}_{\mathrm{S},n}}/{(\partial x_n\,\partial\theta_n)}=2>0$. Hence finite equilibria can be ordered by actions and represented by contiguous type intervals. Sequential rationality of the Receiver gives the conditional-mean action \eqref{eq:best_response_theta}, while the Sender's indifference of adjacent boundary types gives \eqref{eq:arbitrage_cond_theta}. See \Cref{app:thm:general_partition_equilibrium}.
\end{IEEEproof}

Next, we define a \emph{non-informative} equilibrium as a partition equilibrium that corresponds to $\kappa_n=1$. Such a non-informative equilibrium always exists \cite{crawford1982strategic}. The design question is therefore whether an \emph{informative equilibrium} (with $\kappa_n\ge 2$) exists, and if so, what is the maximum number of messages that can be supported at equilibrium, $\kappa_{n,\mathrm{max}} \coloneqq \sup\{\kappa\in\mathbb{N}_{\ge 1} : \text{a $\kappa$-message partition equilibrium exists}\}$.
To study the existence of an informative equilibrium, the following lemma, adapted from \cite{gordon2010on}, is useful.
\begin{lemma}[\cite{gordon2010on}, Thm. 1]
    \label{lemma:existence_monotonicity}
    If an equilibrium with $\kappa_n$ messages exists for Consumer $n$, then an equilibrium exists for all $\kappa_n'<\kappa_n$. Consequently, if no equilibrium exists for $\kappa_n = 2$, none exists for any $\kappa_n > 2$.
\end{lemma}

This allows us to establish a checkable sufficient condition for the existence of an informative equilibrium as follows.

\begin{proposition}[Existence of an Informative Equilibrium]
    \label{prop:informative_delta}
    Let $\mathbb{E}[\theta_n] = \int_{\underline{\theta}_n}^{\overline{\theta}_n} \theta_n \phi_n(\theta_n)\,\mathrm{d}\theta_n$. A sufficient condition for the existence of an informative equilibrium (i.e., $\kappa_n \geq 2$) is
    \begin{equation}
        (2-\gamma_n)\underline{\theta}_n - \gamma_n\mathbb{E}[\theta_n]
        < 2\delta_n <
        (2-\gamma_n)\overline{\theta}_n - \gamma_n\mathbb{E}[\theta_n].
        \label{eq:informative_exist_delta}
    \end{equation}
\end{proposition}
\begin{IEEEproof}
    \IEEEQEDoff
    See \Cref{app:prop:informative_delta}.
\end{IEEEproof}

\subsection{\texorpdfstring{Tariff as an Information-Design Lever}{Tariff as an Information-Design Lever}}
\label{subsec:price_design_lever}

Condition \eqref{eq:informative_exist_delta} shows that informativeness is directly controlled by the pre-announced tariff $p$; the key parameter is the \emph{agreement type} $\theta_n^\star \coloneqq \delta_n/(1-\gamma_n)$, which represents the point of zero bias where the players' interests align ($g_n(\theta_n^\star) = \theta_n^\star$). Mapping this back to Consumer $n$'s original type space gives:
\begin{equation}
    \omega_n^\star \coloneqq \alpha_n\left(\frac{p-b_n^{\mathrm{eff}}}{2a_n^{\mathrm{eff}}}\right) + p .
    \label{eq:omega_star_def_condensed}
\end{equation}
The location of this agreement type relative to Consumer $n$'s type support $\Omega_n=[\underline{\omega}_n, \overline{\omega}_n]$ determines the nature and informativeness of the equilibria as follows.

\begin{proposition}[Equilibrium Regimes]
    \label{prop:agreement_type}
    The tariff $p$ determines the potential for information revelation by shifting the agreement type $\omega_n^\star$ relative to the type support $\Omega_n$:
    \begin{enumerate}[label={\roman*)},wide,ref=\enumref{\theproposition}{\roman*}]
        \item \emph{Outward bias:}
              \label[proposition]{prop:agreement_type:item:1}
              If $\omega_n^\star \in \Omega_n$, the number of messages is unbounded ($\kappa_{n,\mathrm{max}}=\infty$), allowing arbitrarily fine communication near $\omega_n^\star$, though not full revelation.
        \item \emph{One-sided bias:}
              \label[proposition]{prop:agreement_type:2}
              If $\omega_n^\star \notin \Omega_n$, the number of messages is strictly bounded ($\kappa_{n,\mathrm{max}} < \infty$), limiting communication to coarse revelation.
    \end{enumerate}
\end{proposition}
\begin{IEEEproof} % Proof sketch
    % The tariff shifts the zero-bias type $\omega_n^\star$ within the Sender's support. If the agreement type lies inside the support, the affine bias changes sign across the type space, producing the outward-bias case in which equilibria with arbitrarily many intervals can be supported. If it lies outside the support, the bias has one sign throughout the type space, so the standard one-sided-bias boundedness result limits the number of messages. 
    See \Cref{app:prop:agreement_type}.
    \IEEEQEDoff
\end{IEEEproof}

Therefore, by substituting \eqref{eq:omega_star_def_condensed} into \Cref{prop:agreement_type}, the subgame for Consumer $n$ is in the outward-bias regime if and only if $\omega_n^\star\in \Omega_n$, which yields:
\begin{equation}
    \frac{2a_n^{\mathrm{eff}}\,\underline{\omega}_n+\alpha_n b_n^{\mathrm{eff}}}{\alpha_n+2a_n^{\mathrm{eff}}} \;\le\; p \;\le\; \frac{2a_n^{\mathrm{eff}}\,\overline{\omega}_n+\alpha_n b_n^{\mathrm{eff}}}{\alpha_n+2a_n^{\mathrm{eff}}}.
    \label{eq:price_band_outward}
\end{equation}

This formalizes the tariff as an information-design lever: setting $p$ ex ante can shift equilibria toward more messages. Since message count is neither a complete informativeness nor welfare criterion, we next map communication to welfare.

\subsection{\texorpdfstring{Welfare and Value of Communication}{Welfare and Value of Communication}}
\label{subsec:welfare_decomposition}

While Proposition~\ref{prop:agreement_type} shows how the tariff controls the number of supportable messages, message count alone does not capture informativeness or system efficiency. In the
outward-bias regime, even with infinite messages, the intervals may be dense near $\omega_n^\star$,
but coarse elsewhere. What matters for welfare is the posterior uncertainty remaining after communication.

To analyze this, fix a tariff \(p\) and a signaling profile
\(\bm{\sigma}\coloneqq(\sigma_1,\ldots,\sigma_N)\), with
\(m_n=\sigma_n(\omega_n)\). By type independence and local signaling,
\(\widehat{\omega}_n=\mathbb{E}[\omega_n\mid \bm m]
=\mathbb{E}[\omega_n\mid m_n]\). Define
\(\mathbb{E}[u_{\mathrm R}(\bm\sigma)]
\coloneqq
\mathbb{E}[u_{\mathrm R}(\bm\omega,\bm x^*(\widehat{\bm\omega}))]\),
where \(\bm x^*(\widehat{\bm\omega})\) is the Aggregator's Bayes-optimal
dispatch. Let \(u_{\mathrm R}^{\mathrm{FC}}\) and
\(u_{\mathrm R}^{\mathrm{NC}}\) denote the realized welfare under full
communication and no communication, respectively. We show that social welfare decomposes.

\begin{proposition}[Welfare Decomposition]
    \label{prop:welfare_decomposition}
    For a fixed active set \(\mathcal{A}\), fix a tariff \(p\). At every PBE with signaling profile
    \(\bm{\sigma}\):
    \begin{equation}
        \mathbb{E}\!\left[u_{\mathrm{R}}^{\mathrm{FC}}\right]
        -\mathbb{E}\!\left[u_{\mathrm{R}}(\bm{\sigma})\right]
        =
        \sum_{n\in\mathcal{A}}
        c_n\,\mathbb{E}\!\left[
            \mathrm{Var}(\omega_n\mid m_n)\right],
        \label{eq:welfare_decomposition}
    \end{equation}
    where \(c_n\coloneqq
    \tfrac{1}{2}(\alpha_n+2a_n^{\mathrm{eff}})^{-1}>0\).
\end{proposition}

\begin{IEEEproof}
    \IEEEQEDoff
    See \Cref{app:prop_welfare_decomposition}.
\end{IEEEproof}

Thus, the welfare loss relative to full communication is a weighted
sum of the residual posterior variances of the consumers' types. Normalizing \eqref{eq:welfare_decomposition} by the FC-NC gap motivates defining the following \emph{`recovered welfare'} metric.

\begin{corollary}[Recovered Welfare]
    \label{cor:relative_welfare}
    Suppose the dispatch has a fixed active set \(\mathcal A\). If at least one
    active type is nondegenerate, then
    \(\mathbb{E}[u_{\mathrm R}^{\mathrm{FC}}]
    -\mathbb{E}[u_{\mathrm R}^{\mathrm{NC}}]
    =\sum_{n\in\mathcal A}c_n\mathrm{Var}(\omega_n)>0\). Moreover, for every PBE signaling profile \(\bm\sigma\),
    \begin{equation}
        \mathrm{RW}(\bm\sigma)
        \coloneqq
        \frac{
            \mathbb{E}[u_{\mathrm R}(\bm\sigma)]
            -
            \mathbb{E}[u_{\mathrm R}^{\mathrm{NC}}]
        }{
            \mathbb{E}[u_{\mathrm R}^{\mathrm{FC}}]
            -
            \mathbb{E}[u_{\mathrm R}^{\mathrm{NC}}]
        }
        % =
        % \frac{
        %     \sum_{n\in\mathcal A}c_n\mathrm{Var}(\widehat{\omega}_n)
        % }{
        %     \sum_{n\in\mathcal A}c_n\mathrm{Var}(\omega_n)
        % }
        \in[0,1].
        \label{eq:relative_welfare}
    \end{equation}
    Consequently,
    \(\mathbb{E}[u_{\mathrm R}^{\mathrm{NC}}]
    \le \mathbb{E}[u_{\mathrm R}(\bm\sigma)]
    \le \mathbb{E}[u_{\mathrm R}^{\mathrm{FC}}]\).
\end{corollary}

\begin{IEEEproof}
    Under no communication, \(m_n\) is uninformative, so
    \(\widehat{\omega}_n=\mathbb{E}[\omega_n]\) and
    \(\mathbb{E}[\mathrm{Var}(\omega_n\mid m_n)]=\mathrm{Var}(\omega_n)\).
    The FC--NC gap therefore follows from
    \Cref{prop:welfare_decomposition}. For a general profile, the law of total
    variance gives
    \(\mathrm{Var}(\omega_n)
    =
    \mathbb{E}[\mathrm{Var}(\omega_n\mid m_n)]
    +
    \mathrm{Var}(\widehat{\omega}_n)\). Substituting this into
    \eqref{eq:welfare_decomposition} yields
    \(\mathbb{E}[u_{\mathrm R}(\bm\sigma)]
    -\mathbb{E}[u_{\mathrm R}^{\mathrm{NC}}]
    =
    \sum_{n\in\mathcal A}c_n\mathrm{Var}(\widehat{\omega}_n)\), and division by
    the FC--NC gap gives \eqref{eq:relative_welfare}. Bounds follow from
    \(0\le \mathrm{Var}(\widehat{\omega}_n)\le \mathrm{Var}(\omega_n)\).
\end{IEEEproof}

Equation~\eqref{eq:relative_welfare} shows that recovered welfare is the
\(c_n\)-weighted fraction of prior type variance revealed through communication.
We next examine whether informative communication also benefits the reporting consumer.

For Consumer \(n\), let
\(\mathbb{E}[u_{\mathrm{S},n}(\bm{\sigma})]
\coloneqq \mathbb{E}_{\bm{\omega}}\!\left[
    u_{\mathrm{S},n}\bigl(\omega_n,
    x_n^*(\widehat{\bm{\omega}})\bigr)\right]\)
denote its ex ante expected utility under the signaling
profile~\(\bm{\sigma}\).

\begin{corollary}[Consumer Payoff Ranking]
    \label{cor:consumer_payoff_comparison}
    For a fixed active set \(\mathcal{A}\) and tariff \(p\),
    consider Consumer \(n\in\mathcal{A}\) and PBE signaling profiles
    \(\bm{\sigma}, \bm{\sigma}'\) which differ only for Consumer~\(n\).
    Then:
    \begin{equation}
        \begin{aligned}
             & \mathbb{E}[u_{\mathrm{S},n}(\bm{\sigma}')]
            -\mathbb{E}[u_{\mathrm{S},n}(\bm{\sigma})]    \\
             & \quad=c_n(2-\gamma_n)
            \left(\mathrm{Var}(\widehat{\omega}_n')
            -\mathrm{Var}(\widehat{\omega}_n)\right).
        \end{aligned}
        \label{eq:own_information_payoff_gain}
    \end{equation}
    % where \(\gamma_n\in(0,1)\) is the slope of \(g_n\).
    Thus Consumer~\(n\) strictly prefers
    \(\bm{\sigma}'\) over \(\bm{\sigma}\) ex ante whenever
    \(\bm{\sigma}'\) is more informative (in the sense that
    \(\mathrm{Var}(\widehat{\omega}_n') > \mathrm{Var}(\widehat{\omega}_n)\)).
\end{corollary}

\begin{IEEEproof}
    \IEEEQEDoff
    See \Cref{app:cor_consumer_payoff_comparison}.
\end{IEEEproof}
Consequently, in Consumer $n$'s subgame, both the consumer and the Aggregator prefer (ex ante) an informative equilibrium over the non-informative one.

\begin{remark}[Equilibrium Selection]
    \label{rem:equilibrium_selection}
    Equilibrium selection is a classical question in games with
    multiple equilibria \cite{harsanyi1988ageneral}; refinements tailored to cheap-talk games have been proposed~\cite{chen2008selecting}, but a formal analysis of their implications for the present model is beyond the scope of this paper.
    For a PBE signaling profile $\bm{\sigma}$, strictly positive welfare recovery requires informative signaling by at least one active consumer (as non-informative signaling remains a PBE).
    \Cref{cor:consumer_payoff_comparison} shows that informative communication benefits the reporting consumer ex ante, when other consumers' strategies are fixed; but equilibrium selection is not determined by the model.
    Sec.\ref{subsec:price_mechanism_validation} examines welfare sensitivity to alternative equilibrium profiles with heterogeneous message counts.
\end{remark}

\subsection{\texorpdfstring{A Uniform Bias-Minimizing Tariff}{A Uniform Bias-Minimizing Tariff}}
\label{subsec:optimal_price_condensed}

Among the values of tariff $p$ that permit informative communication, which should the Aggregator choose? By \Cref{cor:relative_welfare}, message count alone is not sufficient for welfare: since cells become fine only near $\omega_n^\star$, two tariffs that both permit arbitrarily many messages can yield different welfare.

We therefore adopt a scalar criterion tied directly to the local sender-receiver conflict: minimizing the ex ante mean-squared strategic bias of each effective subgame. For Consumer $n$, define the \emph{mean-squared bias}:

\begin{equation}
    \widebar{B}_n(p)
    \coloneqq
    \mathbb{E}_{\omega_n}
    \left[
        \beta_n
        \!\left(
        \theta_n(\omega_n;p)
        \right)^2
        \right].
    \label{eq:msb_bias_objective}
\end{equation}
% and the bias shifts by $\frac{\partial\beta_n}{\partial p}={-1}/{\alpha_n}$ per unit change in $p$, 
%$\tfrac{\partial\beta_n}{\partial p}=\tfrac{-1}{\alpha_n}$
\noindent As $a_n^{\mathrm{eff}}, b_n^{\mathrm{eff}}$ are $p$-independent and $\tfrac{\mathrm{d}\beta_n}{\mathrm{d} p}=\tfrac{-1}{\alpha_n}$, $\widebar{B}_n(p)$ is strictly convex in $p$, with a unique minimizer satisfying $\mathbb{E}_{\omega_n} \left[ \beta_n \!\left( \theta_n(\omega_n;p) \right) \right] = 0$, or equivalently $\delta_n(p)\!=\!(1-\gamma_n) \mathbb{E}_{\omega_n} \left[ \theta_n(\omega_n;p) \right]$. This gives the \emph{individual bias-minimizing tariff}:
\begin{equation}
    p_n^*
    =
    \frac{\alpha_n b_n^{\mathrm{eff}}
        +2a_n^{\mathrm{eff}}\mathbb{E}[\omega_n]}
    {\alpha_n+2a_n^{\mathrm{eff}}}.
    \label{eq:optimal_price_n}
\end{equation}
The next result shows that these individual minimizers coincide.

\begin{proposition}[Uniform Bias-Minimizing Tariff]
    \label{prop:uniform_optimal_price}
    For any fixed active set $\mathcal{A}$, the unique minimizer of \(\widebar{B}_n(p)\) is independent of the consumer index \(n \in \mathcal{A}\). Hence a single uniform tariff simultaneously minimizes the mean-squared strategic bias of all effective subgames:
    \begin{equation}
        p^*
        =
        \frac{
            b
            +
            2a
            \sum_{j \in \mathcal{A}}
            ({\mathbb{E}[\omega_j]}/{\alpha_j})
        }
        {
            1
            +
            2a
            \sum_{j \in \mathcal{A}}
            (1/\alpha_j)
        } .
        \label{eq:price_minimizing_bias}
    \end{equation}
\end{proposition}

\begin{IEEEproof}
    The first part follows from the strict convexity of $\widebar{B}_n(p)$ and the first-order condition above, which gives \eqref{eq:optimal_price_n}. Substituting the effective parameters from \Cref{cor:equivalent_subgame} into \eqref{eq:optimal_price_n} and multiplying the numerator and denominator by $1+2a\sum_{j \in \mathcal{A}_{-n}}\alpha_j^{-1}$ yields
    \(p_n^*=\frac{
        \alpha_n
        \left(
        b+2a\sum_{j \in \mathcal{A}}
        {\mathbb{E}[\omega_j]}/{\alpha_j}
        \right)
    }{
        \alpha_n
        \left(
        1+2a\sum_{j \in \mathcal{A}}
        {\alpha_j^{-1}}
        \right)
    },
    \)
    $\alpha_n$ cancels, giving \eqref{eq:price_minimizing_bias} which is independent of $n$.
\end{IEEEproof}

% \begin{corollary}[Simultaneous Agreement-Type Alignment]
%    \label{cor:pstar_alignment}
%        For any fixed active set $\mathcal{A}$, let \(\omega_n^\star(p)\) be the agreement-type of Consumer \(n\). At \(p=p^*\), $\omega_n^\star(p^*)=\mathbb{E}[\omega_n]$ for all \(n \in \mathcal{A}\).

%        Equivalently, \(\theta_n^\star=\mathbb{E}[\theta_n]\) for all \(n\). If each prior has positive density on \([\underline{\omega}_n,\overline{\omega}_n]\), then \(\mathbb{E}[\omega_n]\in (\underline{\omega}_n,\overline{\omega}_n)\), and \(p^*\) lies in the interior of every outward-bias price band \eqref{eq:price_band_outward}.
% \end{corollary}
\begin{corollary}[Simultaneous Agreement-Type Alignment]
    \label{cor:pstar_alignment}
    For any fixed active set $\mathcal{A}$, let $\omega_n^\star(p)$ be the agreement type of Consumer $n$. Then:
    \begin{enumerate}[label=(\roman*),wide,ref=\enumref{\thecorollary}{\roman*}]
        \item \label[corollary]{cor:pstar_alignment:i}
              $\omega_n^\star(p^*)=\mathbb{E}[\omega_n]$ (equivalently, $\theta_n^\star=\mathbb{E}[\theta_n]$) for all $n\in\mathcal{A}$,
              and $p^*$ is the unique tariff at which this holds $\forall n \in \mathcal{A}$.
        \item \label[corollary]{cor:pstar_alignment:ii}
              Since $\psi_n$ is strictly positive on $\Omega_n$, $\mathbb{E}[\omega_n]\in(\underline{\omega}_n,\overline{\omega}_n)$.
              Hence $p^*$ lies in the interior of every outward-bias band \eqref{eq:price_band_outward}.
    \end{enumerate}
\end{corollary}

\begin{IEEEproof}
    \IEEEQEDoff
    See \Cref{app:cor_pstar_alignment}.
\end{IEEEproof}
Thus, under this criterion, personalized tariffs do not improve bias alignment over the uniform tariff \(p^*\). The uniform tariff sets each consumer's zero-bias type at its prior mean,    thereby aligning the preferred actions at that type.
% where the region of the partition is most informative.

\begin{remark}[Economic Interpretation of $p^*$]
    The tariff \(p^*\) has a marginal-cost interpretation. Let \(X^{\mathrm{FC}}\) be the aggregate load under the full-communication efficient dispatch. The first-order conditions imply $\omega_n-\alpha_n x_n^{\mathrm{FC}} = C'(X^{\mathrm{FC}})$ for every active consumer. Aggregating these equations and taking expectations gives $p^* = \mathbb{E} \!\left[ C'(X^{\mathrm{FC}}) \right]$. Thus, the bias-minimizing tariff equals the expected efficient marginal operating cost. If priors and costs vary across slots, applying \eqref{eq:price_minimizing_bias} slot by slot yields a day-ahead schedule of bias-minimizing tariffs \(\{p_t^*\}_t\).
\end{remark}

\subsection{Equilibrium Computation}
\label{subsec:computation}
We now describe how to compute a PBE of the subgame. The equilibrium conditions \eqref{eq:best_response_theta}--\eqref{eq:arbitrage_cond_theta} induce an alternating best-response dynamic (BRD) between the Aggregator's dispatch and the consumer's boundary updates.

\emph{1) The Best-Response Map.} First, to evaluate the Aggregator's dispatch, let
\(
\bar{m}_n(y,z)\coloneqq \mathbb{E}[\theta_n\mid y\le\theta_n\le z]
\)
denote the conditional mean of the type on $[y,z]$ (with $\bar{m}_n(y,y)=y$). Since $g_n(\theta_n)=\gamma_n\theta_n+\delta_n$, the induced action in cell $i$ is
\begin{equation}
    x_{n,i}(\bm{\mu}_n) = \gamma_n\bar{m}_n(\mu_{n,i-1},\mu_{n,i})+\delta_n, \quad i=1,\dots,\kappa_n.
    \label{eq:brd_receiver_update}
\end{equation}
Next, combining these actions with the consumer boundary update \eqref{eq:arbitrage_cond_theta} defines the best-response map $T_n$:
\begin{equation}
    [T_n(\bm{\mu}_n)]_i \coloneqq \tfrac{1}{2}\left({x_{n,i}(\bm{\mu}_n)+x_{n,i+1}(\bm{\mu}_n)}\right),
    \label{eq:unprojected_T_def}
\end{equation}
for $i=1,\dots,\kappa_n-1$, and fixed endpoints $\mu_{n,0}=\underline{\theta}_n$, $\mu_{n,\kappa_n}=\overline{\theta}_n$.
Since the input of $T_n$ is a vector of ordered partition boundaries, define, for each integer $\kappa\ge2$, the closed simplex
\begin{equation}
    \begin{aligned}
        \mathcal{D}_n(\kappa)\coloneqq\bigl\{
        \bm{\mu}\in\mathbb{R}^{\kappa-1}:
        \underline{\theta}_n\le\mu_1\le\cdots
        \le\mu_{\kappa-1}\le\overline{\theta}_n
        \bigr\}.
    \end{aligned}
    \label{eq:partition_simplex}
\end{equation}
We write $\mathcal{D}_n\coloneqq\mathcal{D}_n(\kappa_n)$ throughout, and denote by $\mathcal{D}_n^\circ$ its interior, which consists of strictly ordered partitions.
Thus, $T_n:\mathcal{D}_n\to\mathbb{R}^{\kappa_n-1}$, and by \Cref{thm:general_partition_equilibrium}, a strictly ordered partition defines a $\kappa_n$-message PBE if and only if it is a fixed point of $T_n$.

\emph{2) Projected BRD (P-BRD).} For a general tariff $p$, the updates of $T_n$ may leave the type support $[\underline{\theta}_n, \overline{\theta}_n]$. To ensure iterates remain on $\mathcal{D}_n$, we define the coordinatewise projection
\(
\Pi_n(z)\coloneqq
\min\{\overline{\theta}_n,\max\{\underline{\theta}_n,z\}\}
\)
and the projected map
\begin{equation}
    [\widetilde{T}_n(\bm{\mu}_n)]_i
    \coloneqq\Pi_n\bigl([T_n(\bm{\mu}_n)]_i\bigr).
    \label{eq:projected_T_def}
\end{equation}
Because $T_n$ preserves boundary monotonicity and $\Pi_n$ is nondecreasing, $\widetilde{T}_n$ maps $\mathcal{D}_n$ into itself. When the prior density $\phi_n$ is log-concave,\footnote{A probability density function $f$ on an interval is said to be \emph{log-concave} if $\log f$ is a concave function on its support. This class contains the uniform and normal distributions and is closed under truncation to a sub-interval~\cite{bagnoli2005log-concave}; the uniform and truncated-normal priors used in \Cref{sec:simulations} are log-concave.} conditional means are non-expansive in the cell endpoints. Combined with the Receiver's slope $\gamma_n<1$, this makes $\widetilde{T}_n$ a strict contraction, yielding the following convergence and uniqueness guarantees.
\begin{proposition}[Uniqueness and Convergence of P-BRD]
    \label{prop:pbrd_convergence}
    Let $\phi_n$ be continuous, strictly positive on its support, and log-concave. For any finite $\kappa_n\ge2$:
    \begin{enumerate}[label=(\roman*),wide,ref=\enumref{\theproposition}{\roman*}]
        \item \label[proposition]{prop:pbrd_convergence:i} \emph{Contraction \& Convergence:} The projected map satisfies
              \begin{equation}
                  \|\widetilde{T}_n(\bm{\mu}_n)-\widetilde{T}_n(\bm{\nu}_n)\|_\infty
                  \le\gamma_n\|\bm{\mu}_n-\bm{\nu}_n\|_\infty,
                  \label{eq:contraction_display}
              \end{equation}
              for all $\bm{\mu}_n,\bm{\nu}_n\in\mathcal{D}_n$. It has a unique fixed point $\bm{\mu}_n^*\in\mathcal{D}_n$, and the iterates $\bm{\mu}_n^{(\ell)}=\widetilde{T}_n(\bm{\mu}_n^{(\ell-1)})$ converge to $\bm{\mu}_n^*$ at rate $\gamma_n^\ell$ from any initialization in $\mathcal{D}_n$.
        \item \label[proposition]{prop:pbrd_convergence:ii} \emph{Equilibrium Existence \& Uniqueness:} A strictly ordered $\kappa_n$-message PBE exists if and only if
              \begin{equation}
                  \underline{\theta}_n<\mu_{n,1}^*<\cdots
                  <\mu_{n,\kappa_n-1}^*<\overline{\theta}_n.
                  \label{eq:strict_interior_cond}
              \end{equation}
              When \eqref{eq:strict_interior_cond} holds, $\bm{\mu}_n^*$ defines the unique $\kappa_n$-message equilibrium partition and its induced actions. Otherwise, no strictly ordered $\kappa_n$-message PBE exists.
    \end{enumerate}
\end{proposition}
\begin{IEEEproof}
    \IEEEQEDoff
    See \Cref{app:prop_pbrd_convergence}.
\end{IEEEproof}

\emph{3) Numerical Computation.} We implement the P-BRD in \Cref{alg:brd_computation} to calculate the PBE numerically.

For the P-BRD iterates, the distance to the unique fixed point satisfies $\|\bm{\mu}_n^{(\ell)}-\bm{\mu}_n^*\|_\infty\le e_n^{(\ell)}$, where
\begin{align}
    e_n^{(\ell)}\coloneqq\frac{\gamma_n}{1-\gamma_n}
    \|\bm{\mu}_n^{(\ell)}-\bm{\mu}_n^{(\ell-1)}\|_\infty.
    \label{eq:aposteriori_bound}
\end{align}
The bound is derived in \Cref{app:prop_pbrd_convergence} and provides the stopping rule in \Cref{alg:brd_computation}. If every computed cell width exceeds $2e_n^{(\ell)}$, all cells at the fixed point have positive width, certifying the existence of the unique $\kappa_n$-message PBE.

\begin{algorithm}[!t]
    \small
    \caption{Projected BRD for equilibrium computation.}
    \label{alg:brd_computation}
    \begin{algorithmic}[1]
        \Require Message count $\kappa_n\ge2$; support $[\underline{\theta}_n,\overline{\theta}_n]$; prior density $\phi_n$; receiver preference $g_n$; tolerance $\varepsilon>0$; iteration cap $L_{\max}$.
        \State Initialize $\bm{\mu}_n^{(0)}\in\mathcal{D}_n^\circ$.
        \For{$\ell=1,\ldots,L_{\max}$}
        \State $\bm{\mu}_n^{(\ell)}\gets
            \widetilde{T}_n(\bm{\mu}_n^{(\ell-1)})$
        \Comment{Projected-BRD update \eqref{eq:projected_T_def}}
        \State Compute $e_n^{(\ell)}$
        using \eqref{eq:aposteriori_bound}.
        \State \textbf{break if} $e_n^{(\ell)}\le\varepsilon$
        \textbf{and}
        $\min_{1\le i\le\kappa_n}
            (\mu_{n,i}^{(\ell)}-\mu_{n,i-1}^{(\ell)})
            >2e_n^{(\ell)}$
        \EndFor
        \If{$e_n^{(\ell)}\le\varepsilon$ \textbf{and}
            $\min_{1\le i\le\kappa_n}
                (\mu_{n,i}^{(\ell)}-\mu_{n,i-1}^{(\ell)})
                >2e_n^{(\ell)}$}
        \State \Return $\bm{\mu}_n^{(\ell)}$
        \Comment{Equilibrium approximation within $\varepsilon$}
        \Else
        \State \Return \textsc{Unresolved}
        \Comment{Stopping criteria not met.}
        \EndIf
    \end{algorithmic}
\end{algorithm}
\begin{remark}[Message Counts]
    \label{rem:numerical_message_counts}
    We retain the largest tested message count $\kappa_n$ for which a strictly ordered equilibrium partition is verified numerically;
    $\kappa_{n,\mathrm{max}}$ denotes the theoretical supremum.
    An `\textsc{Unresolved}' output means the accuracy and nondegeneracy criteria were not both met within $L_{\max}$ iterations, leaving existence of $\kappa_n$-message equilibrium undecided.
\end{remark}

\begin{remark}[BRD in the Outward-Bias Regime]
    \label{rem:proj_activation}
    In the outward-bias regime, $g_n$ maps the support into itself, so
    $T_n(\mathcal{D}_n)\subseteq\mathcal{D}_n$ and $\widetilde{T}_n=T_n$: the projection is active only outside this regime. In this case, the unprojected dynamic converges directly to the unique $\kappa_n$-message equilibrium.
\end{remark}

\section{Performance Analysis}
\label{sec:simulations}
% \subsection{Setup and Performance Metrics}
% \label{subsec:sim_setup}
We consider $N=18$ consumers and one Aggregator, with utility parameters adapted from the quadratic demand model in \cite{zhao2017consensus-based}. The consumers are grouped by valuation in \Cref{tab:consumers_params}. Group~1 valuations follow a normal distribution $\mathcal{N}(12.5,1.25^2)$ truncated on $[10,15]$.
We compare Strategic Information Transmission (SIT) with Full Communication (FC) and No Communication (NC). Under NC ($\kappa_n=1$ for all $n$), the Aggregator relies only
on the prior densities $\{\psi_n\}_{n=1}^N$.

\Cref{lemma:certificate} classifies Group~0 as fixed-inactive. Unless stated otherwise, experiments use Groups~1--3 as the 15 active-consumer population satisfying condition~\eqref{eq:certificate}, with $C(X)=aX^2+bX$, $a=0.0034$~\$/MWh$^2$, $b=5.56$~\$/MWh, and the resulting bias-minimizing tariff $p^*=9.51$~\$/MWh computed from \eqref{eq:price_minimizing_bias}. All simulations are run with Python on an Intel Core i7-13850HX processor at 3.80 GHz.

\begin{table}[!t]
    \small
    \renewcommand{\arraystretch}{1}
    \renewcommand{\tabcolsep}{0.35em}
    \centering
    \caption{Consumer groups used in the numerical experiments; parameters are adapted from \cite{zhao2017consensus-based}.}
    \label{tab:consumers_params}
    \begin{tabular}{lcccc}
        \toprule
        Group   & $[\underline{\omega}_n,\overline{\omega}_n]$ & Prior         & $\alpha_n$   & Dispatch class     \\
                & (\$/MWh)                                     & $\psi_n$      & (\$/MWh$^2$) & (\# consumers)     \\
        \midrule
        Group 0 & $[5,\,8.5]$                                  & Uniform       & $0.14$       & fixed-inactive (3) \\
        Group 1 & $[10,\,15]$                                  & Trunc. Normal & $0.14$       & active (4)         \\
        Group 2 & $[15,\,18]$                                  & Uniform       & $0.17$       & active (6)         \\
        Group 3 & $[18,\,20]$                                  & Uniform       & $0.19$       & active (5)         \\
        \bottomrule
    \end{tabular}
\end{table}

\noindent The main performance metric is the \emph{recovered welfare} ($\mathrm{RW}$) defined in \eqref{eq:relative_welfare}, and expressed as a percentage: $\mathrm{RW}=0\%$ corresponds to NC and $\mathrm{RW}=100\%$ to FC. Expectations use closed-form expressions or numerical integration in decoupled experiments and Monte Carlo (MC) in joint-game experiments.

\subsection{\texorpdfstring{P-BRD Convergence and Initialization Robustness}{P-BRD Convergence and Initialization Robustness}}
\label{subsec:brd_convergence}
We test the P-BRD in \Cref{alg:brd_computation}, using the error bound $e_n^{(\ell)}$ in \eqref{eq:aposteriori_bound}. Using Group~1's parameters, we vary the population size $N\in\{1,5,10,20,30,40\}$, fix $\kappa_n=20$, and use $a=0.1$~\$/MWh$^2$ with tolerance $\varepsilon=10^{-4}$. We recompute $p^*$ for each $N$. The error bound in \Cref{fig2:convergenceBRD} falls below
$\varepsilon$ for every tested $N$.
Convergence slows as $\gamma_n\to1$, but remains stable.

\begin{figure}[htbp]
    \centering
    \includegraphics[width=0.85\columnwidth]{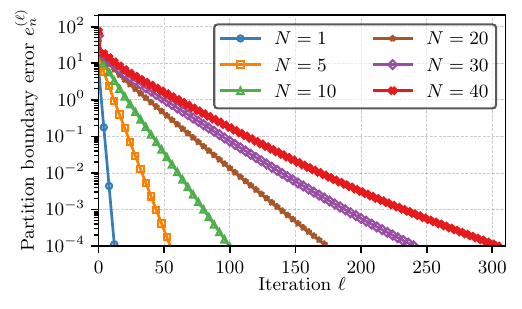}
    \caption{P-BRD convergence. The error bound falls below $\varepsilon=10^{-4}$ for every tested population size $N$. All runs satisfy both the accuracy and nondegeneracy criteria of \Cref{alg:brd_computation} within the iteration cap $L_{\max}=20{,}000$.}
    \label{fig2:convergenceBRD}
\end{figure}

We also test initialization sensitivity with $a=0.0034$~\$/MWh$^2$ and $p^*\approx9.51$~\$/MWh, computed from \eqref{eq:price_minimizing_bias}.
To evaluate P-BRD across equilibrium regimes, we test $p^*$ and $1.01p^*$ to cover outward and one-sided bias,
including runs with active projection (\Cref{rem:proj_activation}).
For each group and tariff, we test with the verified count $\kappa_{\mathrm{tested}}$ (capped at $20$) listed in \Cref{tab2:brd_robustness}.
For each row, we run \Cref{alg:brd_computation} from an equally spaced reference partition and from 20 random initializations. All runs converged to a strictly ordered partition, with a maximum spread below $2\varepsilon$ and a total elapsed time of 1.93 seconds.

\begin{table}[t]
    \centering
    \small
    \setlength{\tabcolsep}{3pt}
    \caption{P-BRD initialization robustness: maxima across initializations. Spread is measured from an equally spaced initialization reference.}
    \label{tab2:brd_robustness}
    \begin{tabular}{llcccc}
        \toprule
        Tariff    & Group   & $\kappa_{\mathrm{tested}}$ & \shortstack{Max\\spread} & \shortstack{Max\\iters} & \shortstack{Max proj.\\steps} \\
        \midrule
        $p^*$     & Group 1 & $20$                       & $1.21\times10^{-4}$      & $347$                   & $0$                           \\
                  & Group 2 & $20$                       & $1.11\times10^{-4}$      & $379$                   & $0$                           \\
                  & Group 3 & $20$                       & $1.06\times10^{-4}$      & $397$                   & $0$                           \\
        \midrule
        $1.01p^*$ & Group 1 & $9$                        & $7.91\times10^{-5}$      & $195$                   & $63$                          \\
                  & Group 2 & $5$                        & $3.60\times10^{-5}$      & $109$                   & $15$                          \\
                  & Group 3 & $4$                        & $1.97\times10^{-5}$      & $76$                    & $11$                          \\
        \bottomrule
    \end{tabular}
\end{table}

\subsection{\texorpdfstring{Tariff, Information, and Welfare}{Tariff, Information, and Welfare}}
\label{subsec:price_mechanism_validation}
We vary the tariff in $[0.8p^*,1.2p^*]$ for the 15-consumer active cohort, with message counts capped at $20$. For Group~1, only a narrow interval around $p^*$ supports $\kappa_{n,\mathrm{max}}=\infty$ (outward bias); outside it, communication becomes coarser (one-sided bias), as shown in \Cref{fig3:price_mechanism}(a).
\Cref{fig3:price_mechanism}(b) shows that $p^*$ minimizes the aggregate mean-squared bias, $B(p)\coloneqq\sum_{n\in\mathcal{A}}\widebar{B}_n(p)$ (\Cref{prop:uniform_optimal_price}), and yields the highest recovered welfare among the tested tariffs ($98.9\%$).
\Cref{fig4:info_selection}(b) relates recovered welfare to mean residual posterior variance $\mathbb{E}[\mathrm{Var}(\omega_n\mid m_n)]$ in $(\$/\mathrm{MWh})^2$. As the message count rises from $\kappa=2$ to $\kappa=12$, recovered welfare increases from $71.3\%$ to $98.6\%$ while residual variance decreases (\Cref{prop:welfare_decomposition}).

\emph{Sensitivity to Equilibrium Selection:}~At $p^*$, consumers may use heterogeneous message counts. For the default 15-consumer active cohort, we sample 400 equilibrium profiles by drawing each consumer's message count independently from $\kappa_n\in\{1,\dots,12\}$. This sampling evaluates welfare sensitivity, but does not model how consumers would select an equilibrium. Under this scheme, over $75\%$ of the profiles recover more than $80\%$ of the FC--NC gain, with a median of $87.9\%$ (\Cref{fig4:info_selection}(a)). Every sampled profile (including heterogeneous message counts) has nonnegative $\mathrm{RW}$, consistent
with \Cref{cor:relative_welfare}.

\begin{figure}[htbp]
    \centering
    \includegraphics[width=0.85\columnwidth,trim=0 1.5mm 0 1.5mm,clip]{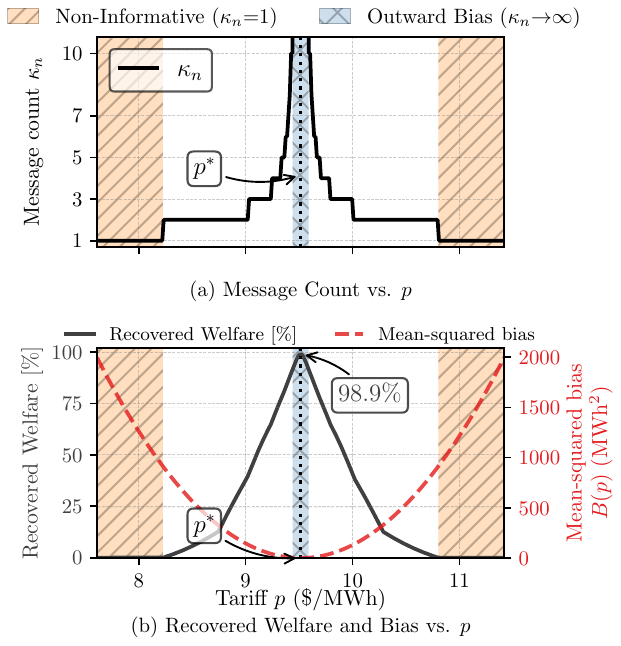}
    \caption{Tariff effects on informativeness and welfare. (a) Group~1 message count: finite $\kappa$ are capped at $20$ (displayed up to $10$); unbounded $\kappa$ are identified via Prop.~\ref{prop:agreement_type:item:1}. Shading denotes non-informative (orange hatched) and outward-bias (blue) regimes. (b) Recovered welfare and aggregate mean-squared bias $B(p)$ (MWh$^2$); at $p^*$, the bias is minimized and $\mathrm{RW}$ peaks.}
    \label{fig3:price_mechanism}
\end{figure}

\begin{figure}[!t]
    \centering
    \includegraphics[width=0.9\columnwidth]{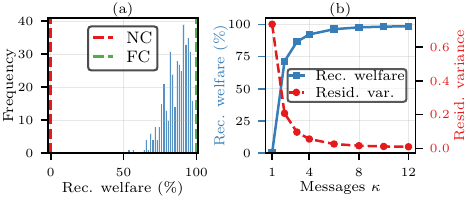}
    \caption{At $p=p^*$: (a) Recovered welfare over $400$ random equilibrium profiles. (b) Recovered welfare and mean residual variance versus message count $\kappa$.}
    \label{fig4:info_selection}
\end{figure}

\subsection{\texorpdfstring{Strategic Decoupling Validation}{Strategic Decoupling Validation}}
\label{subsec:decoupling_validation}
\emph{Test A (Payoff-difference invariance):} By \Cref{thm:strategic_independence}, other consumers' strategies affect an active consumer's expected payoff only through an additive constant, so payoff differences across messages are invariant to their signaling profiles.
We test this on a 5-consumer active core (two Group~1, two Group~2, one Group~3), evaluating a focal Group~1 consumer under four strategy profiles of the other consumers ($\kappa_j\in\{1,2,12\}$ and mixed counts of $\{3,5\}$). For each profile and consumer type, we calculate the expected payoff difference between each message and a fixed reference message. Using $M$ MC draws, we define $\Delta(M)$ as the largest change in these payoff differences relative to the reference strategy profile, across the evaluated types and posterior means.
Without decoupling, these payoff differences may depend on others' signaling strategies and potentially change the best response. For the fixed-active case, \Cref{fig5:decoupling_validation}(a) (blue) shows that $\Delta(M)$ generally decreases with more MC draws, reaching $1.32\times10^{-3}$ at $M=10^6$, which is consistent with the predicted payoff-difference invariance.
\emph{Test A2:} We repeat Test A when \eqref{eq:certificate} breaks by adding three Group~0 consumers, with a focal Group~3 consumer. Under the prescribed profiles, the active set becomes message-dependent. Here, $\Delta(M)$ does not decay, and remains near $\$0.28$ across the tested $M$ (\Cref{fig5:decoupling_validation}(a), red), showing the strategic coupling.

\emph{Test B (Coupled-dispatch welfare):} For the same 5-consumer core ($\kappa_n=12$), \Cref{fig5:decoupling_validation}(b) shows that the constrained-dispatch welfare closely matches the decoupled prediction. The MC estimates ($\pm1$ SE) are $\$1286.97\pm\$0.68$ at $M=2\times10^4$ and $\$1286.55\pm\$0.14$ at $M=5\times10^5$, versus $\$1286.66$ analytically.
% The observed constraint-binding rate is $0\%$, consistent with \Cref{lemma:certificate}.

\begin{figure}[!t]
    \centering
    \includegraphics[width=0.9\columnwidth]{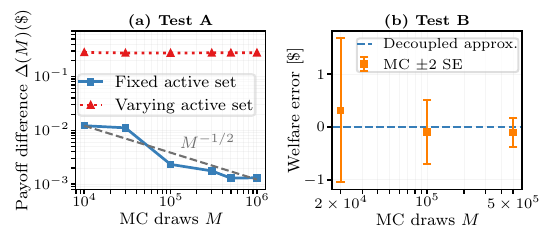}
    \caption{Strategic decoupling diagnostics. (a) Maximum cross-profile deviation of normalized payoff maps for fixed and varying active sets. (b) MC welfare minus the analytical prediction, with pointwise $\pm2$ SE intervals.}
    \label{fig5:decoupling_validation}
\end{figure}

\subsection{Robustness to Participation and Tariff-Design Error}
\label{subsec:selection_robustness}

\emph{1) Participation Regimes:}
Decoupling requires a fixed active set (\Cref{thm:strategic_independence}). Marginal valuations can make dispatch non-negativity constraints bind stochastically, leading to varying participation (\Cref{def:participation_regimes:iii}). We test this transition with the 5-consumer core from \Cref{subsec:decoupling_validation}, adding three Group~0 consumers with shifted support $\omega_n\sim\mathcal{U}[5+s,8.5+s]$. We vary $s$ from $+4.5$ to $-4.5$~\$/MWh in $0.25$~\$/MWh steps, using $4\times10^4$ MC draws per step. The shift $s$ spans full participation ($s\ge3.0$, \eqref{eq:certificate}), varying participation, and partial participation ($s\le-1.75$, \eqref{eq:two_tier_C1}--\eqref{eq:two_tier_C2}), as shown in \Cref{fig:marginal}.
The Aggregator constrains dispatch to be nonnegative, and sets the tariff using \eqref{eq:price_minimizing_bias} after excluding consumers certified as fixed-inactive (by \Cref{lemma:certificate}).
The decoupled prediction lies within the pointwise MC error bands in both full and partial regimes (\Cref{fig:marginal}).
Under varying participation, the discrepancy in predicted $\mathrm{RW}$ peaks at $0.52$ percentage points.
We emphasize that this comparison evaluates welfare under the prescribed decoupled strategy profile rather than at an equilibrium of the coupled game; in the varying-participation regime it is therefore an empirical welfare comparison, not an equilibrium result.

Additionally, to assess equilibrium proximity, we also report the largest ex ante gain that any consumer can obtain by unilaterally deviating over the prescribed message menus (\Cref{fig:marginal} purple line). It stays below \$$10^{-5}$ wherever a fixed active set is certified (left and right bands in \Cref{fig:marginal}), i.e., numerically zero, consistent with \Cref{thm:strategic_independence}. It remains below \$$0.15$ per DR event (which corresponds to $0.13\%$ of the deviating consumer's expected payoff under the prescribed profile) across the varying regime, rising to its maximum near the certification boundary.

\begin{figure}[!t]
    \centering
    \includegraphics[width=0.8\columnwidth]{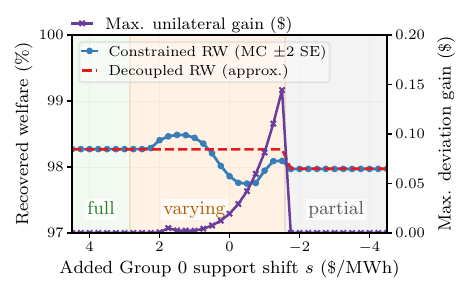}
    \caption{Constrained vs. decoupled recovered welfare and maximum unilateral-deviation gain under valuation shift $s$ for the added Group~0. Vertical bands mark the full, varying, and partial-participation regimes of~\Cref{def:participation_regimes}; the discrepancy peaks at $\approx0.52$ percentage points in the varying regime and is within MC error at both full and partial fixed-active-set regimes.}
    \label{fig:marginal}
\end{figure}

\begin{figure}[!t]
    \centering
    \includegraphics[width=0.85\columnwidth]{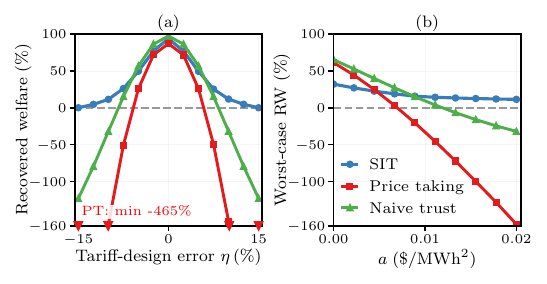}
    \caption{Comparative performance under tariff-design error $\eta$. (a) Recovered welfare across $\pm15\%$ tariff error at $a=0.02$~\$/MWh$^2$. Panel (a) truncates $\mathrm{RW}$ at $-160\%$; downward triangles mark continuation below the axis.
        (b) Worst-case recovered welfare over $\pm 10\%$ tariff error as $a$ increases.}
    \label{fig7:misspec}
\end{figure}

\emph{2) Comparison with Baselines Under Tariff-Design Error:}
We model tariff-design error as $p=(1+\eta)p^*$, with $\eta\in[-15\%,+15\%]$, using $3\times10^4$ MC draws. We compare SIT ($\kappa_n\le20$) with two benchmarks on a 3-consumer cohort (one from each of Groups~1--3), with $a=0.02$~\$/MWh$^2$:

\begin{itemize}[leftmargin=*,wide]
    \item \emph{Price Taking (PT):} A decentralized benchmark without communication, where consumers self-dispatch at tariff $p$ according to $x_n^{\mathrm{PT}}=\max\{0,(\omega_n-p)/\alpha_n\}$.
          % \item \emph{Naive Trust (NT):} A direct reporting mechanism in which the Aggregator treats reports as truthful and dispatches via \eqref{eq:Aggregator_best_response}. Consumers anticipate this rule and choose optimal misreports clipped to $[\underline{\omega}_n,\overline{\omega}_n]$.
    \item \emph{Naive Trust (NT):} A direct reporting mechanism where the Aggregator treats reported parameters $\tilde{\bm{\omega}}$ as truthful and dispatches $x_n^{\mathrm{NT}} = \max\{0, x_{n,\mathcal{V}}^{\mathrm{unc}}(\tilde{\bm{\omega}})\}$ via \eqref{eq:Aggregator_best_response}. Anticipating this, Consumer~$n$ sends the optimal misreport (which maximizes its expected net payoff in \eqref{eq:consumer_net_utility}) clipped to $[\underline{\omega}_n, \overline{\omega}_n]$. %$\bigl(1 + \tfrac{2a_n^{\mathrm{eff}}}{\alpha_n}\bigr)(\omega_n - p) + b_n^{\mathrm{eff}}$
\end{itemize}

Here only the tariff is perturbed; signaling equilibria
use the true priors. Since $p^*$ depends on the prior means through \eqref{eq:price_minimizing_bias}, misspecification of those means can induce such tariff-design errors.
At $p^*$, NT and PT achieve $\mathrm{RW}$ values of $97.4\%$ and $87.0\%$, respectively, versus $93.8\%$ for SIT (\Cref{fig7:misspec}(a)). Tariff errors distort PT consumption and NT reports, with minimum sampled $\mathrm{RW}$ values of $-465\%$ and $-123\%$, respectively. SIT uses $p$ only as an information-design parameter: tariff errors change the selected reporting partitions, including their message counts. Its $\mathrm{RW}$ remains nonnegative across all tested errors, as \Cref{cor:relative_welfare} guarantees.
\Cref{fig7:misspec}(b) reports the worst-case recovered welfare over $\pm10\%$ perturbations with $a\in[10^{-5},0.02]$~\$/MWh$^2$. At the largest tested curvature $a$, PT and NT fall below NC ($-158\%$ and $-32\%$ RW, respectively), while SIT recovers $11\%$ of the FC--NC gain and remains positive at all tested curvatures.
% These results highlight SIT's robustness to tariff-design error.

\section{Conclusion}
\label{sec:conclusion}
This paper developed a tractable framework for strategic communication in demand response with heterogeneous consumers. Under explicit conditions ensuring a fixed dispatched set of consumers, reporting incentives decouple into independent subgames despite coupled dispatch. A closed-form uniform tariff simultaneously minimizes all active consumers' mean-squared strategic bias. For log-concave priors, each feasible finite message count admits a unique equilibrium partition, computable by a globally convergent algorithm.
Welfare analysis establishes that every equilibrium weakly improves expected social welfare over no communication within this regime. Numerical studies show substantial welfare recovery under selected informative equilibria, and nonnegative recovery under tested tariff-design errors. These findings provide a tractable basis for using retail pricing to improve information revelation in DR. Future extensions may address repeated interactions with intertemporal coupling, incorporate network constraints, examine designs with commitment, and explore more general payoff structures.

\appendices
\crefalias{section}{appendix}
\crefalias{subsection}{appendix}

\section{Aggregator's Best Response}
\label{app:Aggregator_br_derivation}
Throughout, for any $\mathcal{S}\subseteq\mathcal{V}$, let $\mathcal{G}_{\mathcal{S}}\coloneqq 1+2a\sum_{j\in\mathcal{S}}\alpha_j^{-1}$.
\begin{IEEEproof}
    The Aggregator's problem \eqref{eq:pbe_receiver} maximizes the expected social welfare $\mathbb{E}[u_{\mathrm{R}} \mid \bm{m}]$ subject to $\bm{x} \ge 0$. This objective is quadratic in $\bm{x}$ with Hessian $\bm{H} = -\mathrm{diag}(\bm{\alpha}) - 2a \bm{1}\bm{1}^\top$, where $\bm{1}$ is the all-ones vector. Since the parameters $\alpha_n, a$ are strictly positive, $\bm{H}$ is strictly negative definite. The problem is therefore strictly concave, guaranteeing that the Karush-Kuhn-Tucker (KKT) conditions uniquely identify the global maximum. Relaxing the non-negativity constraints, the first-order condition (FOC) with respect to $x_n$ is:
    \begin{equation}
        \frac{\partial \mathbb{E}[u_{\mathrm{R}} \mid \bm{m}]}{\partial x_n} = \widehat{\omega}_n - \alpha_n x_n - 2aX - b = 0,
    \end{equation}
    where $X = \sum_{j\in\mathcal{V}} x_j$. Rearranging yields $x_n = (\widehat{\omega}_n - b - 2aX)/\alpha_n$. Summing this over all consumers $j \in \mathcal{V}$ gives:
    \begin{equation}
        X = \sum_{j\in\mathcal{V}} \frac{\widehat{\omega}_j - b}{\alpha_j} - 2a X \sum_{j\in\mathcal{V}} \alpha_j^{-1}.
        \label{eq:app1-aggregate-load}
    \end{equation}
    Factoring out $X$ and applying the definition of $\mathcal{G}_{\mathcal{V}}$, this simplifies to $X \mathcal{G}_{\mathcal{V}} = \sum_{j\in\mathcal{V}} (\widehat{\omega}_j - b)/\alpha_j$. Substituting $X = \mathcal{G}_{\mathcal{V}}^{-1} \sum_{j\in\mathcal{V}} (\widehat{\omega}_j - b)/\alpha_j$ back into the expression for $x_n$ recovers the unconstrained dispatch rule \eqref{eq:Aggregator_best_response}:
    \begin{equation}
        x_n^{\mathrm{unc}}(\widehat{\bm{\omega}}) = \frac{\widehat{\omega}_n - b}{\alpha_n} - \frac{2a}{\alpha_n \mathcal{G}_{\mathcal{V}}} \sum_{j\in\mathcal{V}} \frac{\widehat{\omega}_j - b}{\alpha_j}.
    \end{equation}
    When $x_n^{\mathrm{unc}}(\widehat{\bm{\omega}}) \ge 0$ for all $n$ (e.g., under full participation certified by condition \eqref{eq:certificate}), the unconstrained solution is feasible and is the unique optimal dispatch.
\end{IEEEproof}

\section{\texorpdfstring{Proof of \Cref{lemma:certificate}}{Proof of full-participation lemma}}
\label{app:lemma_certificate}
\begin{IEEEproof}
    By the strict concavity established in Appendix~\ref{app:Aggregator_br_derivation}, a candidate split $(\mathcal{A}, \mathcal{I})$ coincides with the split induced by the optimal dispatch $\bm{x}^*(\widehat{\bm{\omega}})$ (i.e., $\mathcal{A}(\widehat{\bm{\omega}}) = \mathcal{A}$ and $\mathcal{I}(\widehat{\bm{\omega}}) = \mathcal{I}$) across all admissible posteriors $\widehat{\bm{\omega}} \in \bm{\Omega}$ if and only if two independent requirements hold simultaneously: no active consumer is driven to zero ($x_n > 0, \forall n \in \mathcal{A}$), \emph{and} no inactive consumer warrants a positive dispatch ($x_i = 0, \forall i \in \mathcal{I}$).
    % We establish the exact bounds for these two requirements.

    \emph{1) Inactive requirement \eqref{eq:two_tier_C2}:} For $i \in \mathcal{I}$, setting $x_i = 0$ in the KKT stationarity condition requires $\widehat{\omega}_i \le 2a X_{\mathcal{A}} + b$, where $X_{\mathcal{A}}$ is the optimal aggregate load of the active set. For this to hold across the entire belief space $\bm{\Omega}$, it must hold in the worst-case scenario where $\widehat{\omega}_i$ is maximized ($\overline{\omega}_i$) and $X_{\mathcal{A}}$ is minimized. Because $X_{\mathcal{A}}$ is strictly increasing in active valuations, its minimum occurs at $\underline{\bm{\omega}}_{\mathcal{A}}$.
    % Substituting this into the aggregate load formula \eqref{eq:app1-aggregate-load} yields condition \eqref{eq:two_tier_C2}:
    The aggregate load is generated by $\mathcal{A}$ alone, and \eqref{eq:app1-aggregate-load} applies with $\mathcal{V}$ replaced by $\mathcal{A}$.
    % i.e., $X_{\mathcal{A}}=\mathcal{G}_{\mathcal{A}}^{-1}\sum_{j\in\mathcal{A}}(\widehat{\omega}_j-b)/\alpha_j$
    Evaluating at $\underline{\bm{\omega}}_{\mathcal{A}}$ yields condition \eqref{eq:two_tier_C2}:

    \begin{equation*}
        \overline{\omega}_i \le 2a X_{\mathcal{A}}(\underline{\bm{\omega}}_{\mathcal{A}}) + b = b + \frac{2a}{\mathcal{G}_{\mathcal{A}}} \sum_{j\in\mathcal{A}} \frac{\underline{\omega}_j - b}{\alpha_j}.
    \end{equation*}
    \emph{2) Active requirement \eqref{eq:two_tier_C1}:} For $n \in \mathcal{A}$, the dispatch must remain strictly positive. Using the unconstrained dispatch rule \eqref{eq:Aggregator_best_response} applied over the active set $\mathcal{A}$:
    \begin{equation*}
        x_n = \frac{\widehat{\omega}_n - b}{\alpha_n} - \frac{2a}{\alpha_n \mathcal{G}_{\mathcal{A}}} \sum_{j \in \mathcal{A}} \frac{\widehat{\omega}_j - b}{\alpha_j} > 0.
    \end{equation*}
    Multiplying by $\alpha_n \mathcal{G}_{\mathcal{A}}$ and decomposing the scaling factor as $\mathcal{G}_{\mathcal{A}} = \mathcal{G}_{\mathcal{A}_{-n}} + 2a\alpha_n^{-1}$ yields:
    \begin{equation*}
        (\widehat{\omega}_n - b)\left( \mathcal{G}_{\mathcal{A}_{-n}} + \frac{2a}{\alpha_n} \right) - \frac{2a}{\alpha_n}(\widehat{\omega}_n - b) - 2a \sum_{j \in \mathcal{A}_{-n}} \frac{\widehat{\omega}_j - b}{\alpha_j} > 0.
    \end{equation*}
    The terms scaled by $2a/\alpha_n$ cancel. Dividing by $\mathcal{G}_{\mathcal{A}_{-n}}$ gives:
    \begin{equation*}
        \widehat{\omega}_n > b + \frac{2a}{\mathcal{G}_{\mathcal{A}_{-n}}} \sum_{j \in \mathcal{A}_{-n}} \frac{\widehat{\omega}_j - b}{\alpha_j}.
    \end{equation*}
    For this strict inequality to hold unconditionally, it must be satisfied in the worst-case scenario where $\widehat{\omega}_n$ is minimized ($\underline{\omega}_n$) and the active valuations are maximized ($\overline{\bm{\omega}}_{\mathcal{A}_{-n}}$). Substituting these worst-case bounds yields condition \eqref{eq:two_tier_C1}. Therefore, the partition is fixed if and only if \eqref{eq:two_tier_C1} and \eqref{eq:two_tier_C2} hold.
\end{IEEEproof}

\section{\texorpdfstring{Proof of \Cref{thm:strategic_independence}}{Proof of strategic-independence theorem}}
\label{app:thm_two_tier}

\begin{IEEEproof}
    We show that for any active Consumer $n \in \mathcal{A}$, the expected-utility ranking over \emph{all} messages is invariant to the strategy profile $(\sigma_j)_{j\in\mathcal{A}_{-n}}$ of the other active consumers. When the dispatch has a fixed active set $\mathcal{A}$, the dispatch for active consumers follows the affine rule $x_n = x_{n,\mathcal{A}}^{\mathrm{unc}}(\widehat{\bm{\omega}})$ given in \eqref{eq:active_dispatch_short}. Fix any message $m_n \in \mathbb{M}$ and let $\widehat{\omega}_n = \widehat{\omega}_n(m_n)$ denote its induced posterior mean. We decompose $x_n$ as:
    % (for off-path messages, any admissible mean in $\Omega_n$)
    % We decompose $x_n$ as:
    \begin{equation*}
        x_n = \underbrace{\left( \frac{1}{\alpha_n} - \frac{2a}{\alpha_n^2 \mathcal{G}_{\mathcal{A}}} \right)(\widehat{\omega}_n-b)}_{D_n(\widehat{\omega}_n)} \; - \underbrace{\frac{2a}{\alpha_n \mathcal{G}_{\mathcal{A}}} \sum_{j \in \mathcal{A}_{-n}} \frac{\widehat{\omega}_j-b}{\alpha_j}}_{R_n}.
    \end{equation*}
    Here, $D_n(\widehat{\omega}_n)$ is deterministic given Consumer $n$'s own message, while, from Consumer $n$'s perspective, $R_n$ is random (through $\bm{\omega}_{-n}$ and $\{\sigma_j\}_{j \in \mathcal{A}_{-n}}$) but does not depend on $m_n$. Expanding the expected utility \eqref{eq:consumer_net_utility} using $\mathbb{E}[Z^2] = (\mathbb{E}[Z])^2 + \mathrm{Var}(Z)$, the expected payoff of type $\omega_n$ from a message inducing $\widehat{\omega}_n$ is
    \begin{equation}
        \mathbb{E}_{\bm{\omega}_{-n}}\!\left[u_{\mathrm{S},n}\right]
        = (\omega_n - p)\,\mathbb{E}[x_n]
        - \tfrac{\alpha_n}{2} \left[ (\mathbb{E}[x_n])^2 + \mathrm{Var}(x_n) \right],
        \label{eq:expected_payoff_decomp}
    \end{equation}
    where $\mathbb{E}[x_n] = D_n(\widehat{\omega}_n) - \mathbb{E}[R_n]$ and $\mathrm{Var}(x_n) = \mathrm{Var}(R_n) \coloneqq V_n$. Two observations conclude the proof: (i) $V_n$ does not depend on $m_n$, so $-\frac{\alpha_n}{2}V_n$ acts as an additive constant common to \emph{all} messages for Consumer $n$; (ii) by the law of iterated expectations (LIE), $\mathbb{E}[\widehat{\omega}_j] = \mathbb{E}[\omega_j]$ for every $j \in \mathcal{A}_{-n}$ \emph{regardless of} their signaling strategies $\sigma_j$. Consequently, $\mathbb{E}[R_n] = \frac{2a}{\alpha_n \mathcal{G}_{\mathcal{A}}}\sum_{j\in\mathcal{A}_{-n}}\frac{\mathbb{E}[\omega_j]-b}{\alpha_j}$ depends on the other active consumers only through their fixed prior means.

    Hence, up to a message-independent constant, the map $(\widehat{\omega}_n, \omega_n) \mapsto \mathbb{E}_{\bm{\omega}_{-n}}[u_{\mathrm{S},n}]$ is a fixed function of the model parameters alone. This means Consumer $n$'s best-response problem, and with it its entire set of equilibrium strategies, is identical across all strategy profiles of the other consumers. In particular, the condition for a boundary type $\tau$ to be indifferent between two distinct messages inducing different expected dispatch $\mathbb{E}[x_n] \neq \mathbb{E}[x_n']$ simplifies via \eqref{eq:expected_payoff_decomp}. Equating the expected payoffs, the common variance term $V_n$ cancels out, yielding $(\tau - p)(\mathbb{E}[x_n] - \mathbb{E}[x_n']) = \tfrac{\alpha_n}{2}\big((\mathbb{E}[x_n])^2 - (\mathbb{E}[x_n'])^2\big)$. Dividing by $\mathbb{E}[x_n] - \mathbb{E}[x_n']$ reduces this to $\tau - p = \frac{\alpha_n}{2}\big(\mathbb{E}[x_n] + \mathbb{E}[x_n']\big)$, which involves only the type, price, and expected dispatch constants determined entirely by the prior distributions. This establishes strategic decoupling between consumers on $\mathcal{A}$.
\end{IEEEproof}

\section{\texorpdfstring{Proof of \Cref{cor:equivalent_subgame}}{Proof of equivalent-subgame corollary}}
\label{app:proof_corollary}

\begin{IEEEproof}
    We show that the expected dispatch and expected payoffs in the multi-sender game reduce to the defined effective single-sender game, up to a message-independent~constant.

    In the effective single-sender game, the Receiver chooses $x_n$ to maximize the effective utility \eqref{eq:receiver_eff_util} based on the belief $\widehat{\omega}_n$. The FOC yields the deterministic optimal dispatch:
    \begin{equation}
        x_n^{\mathrm{eff},*}(\widehat{\omega}_n) = \frac{\widehat{\omega}_n - b_n^{\mathrm{eff}}}{\alpha_n + 2a_n^{\mathrm{eff}}}.
        \label{eq:proof_br_eq}
    \end{equation}

    In the global multi-sender game, Consumer $n$'s allocation is $x_{n,\mathcal{A}}^{*}$ given by \eqref{eq:active_dispatch_short}. Taking the expectation over the other consumers' reports, and using LIE, $\mathbb{E}[\widehat{\omega}_j] = \mathbb{E}[\omega_j]$ for $j \in \mathcal{A}_{-n}$, we compute the expected allocation $\mathbb{E}[x_{n}^{*}]$.
    Let $\mathcal{H}_n \coloneqq \sum_{j \in \mathcal{A}_{-n}} (\mathbb{E}[\omega_j]-b)/\alpha_j$. Then:
    \begin{align*}
        \mathbb{E}[x_{n}^{*}] & = \frac{\widehat{\omega}_n-b}{\alpha_n} - \frac{2a}{\alpha_n} \frac{\frac{\widehat{\omega}_n-b}{\alpha_n} + \mathcal{H}_n}{\mathcal{G}_{\mathcal{A}_{-n}} + 2a/\alpha_n}                                                              \\
                              & = \frac{(\widehat{\omega}_n-b)\big(\mathcal{G}_{\mathcal{A}_{-n}}+\frac{2a}{\alpha_n}\big) - 2a\big(\frac{\widehat{\omega}_n-b}{\alpha_n} + \mathcal{H}_n\big)}{\alpha_n\big(\mathcal{G}_{\mathcal{A}_{-n}}+\frac{2a}{\alpha_n}\big)} \\
                              & = \frac{(\widehat{\omega}_n-b)\mathcal{G}_{\mathcal{A}_{-n}} - 2a \mathcal{H}_n}{\alpha_n\mathcal{G}_{\mathcal{A}_{-n}} + 2a}.
    \end{align*}
    The second line places both terms over the common denominator; the terms $\pm\frac{2a}{\alpha_n}(\widehat{\omega}_n - b)$ in the numerator then cancel, yielding the third line. Dividing the numerator and denominator by $\mathcal{G}_{\mathcal{A}_{-n}}$ yields:
    \begin{equation}
        \label{eq:exp-optimal-action-proof}
        \mathbb{E}[x_{n}^{*}] = \frac{\widehat{\omega}_n - \left(b + \frac{2a \mathcal{H}_n}{\mathcal{G}_{\mathcal{A}_{-n}}}\right)}{\alpha_n + \frac{2a}{\mathcal{G}_{\mathcal{A}_{-n}}}}.
    \end{equation}
    Comparing \eqref{eq:exp-optimal-action-proof} with \eqref{eq:proof_br_eq} shows that $\mathbb{E}[x_{n}^{*}] = x_n^{\mathrm{eff},*}(\widehat{\omega}_n)$ under the definitions of $a_n^{\mathrm{eff}}$ and $b_n^{\mathrm{eff}}$ stated in \Cref{cor:equivalent_subgame}.

    Finally, we establish strategic payoff-equivalence. As detailed in the proof of \Cref{thm:strategic_independence}, Consumer $n$'s expected utility in the multi-sender game expands to:
    \begin{align*}
        \mathbb{E}_{\bm{\omega}_{-n}}\!\left[u_{\mathrm{S},n}\right]
         & = (\omega_n - p)\,\mathbb{E}[x_n^*] - \frac{\alpha_n}{2}\Big( \big(\mathbb{E}[x_n^*]\big)^2 + \mathrm{Var}(x_n^*) \Big) \\
         & = u_{\mathrm{S},n}(\omega_n, x_n^{\mathrm{eff},*}) - \frac{\alpha_n}{2}\mathrm{Var}(x_n^*),
    \end{align*}
    where $u_{\mathrm{S},n}(\omega_n, x_n^{\mathrm{eff},*})$ is the utility evaluated directly within the effective single-sender game. Because the variance term $\mathrm{Var}(x_n^*)$ depends only on the prior distributions and signaling strategies of other consumers (and is independent of Consumer $n$'s own message), the expected payoffs in the two games differ only by an additive, message-independent constant. Therefore, both games induce the exact same preference ranking over messages for Consumer $n$, establishing payoff-equivalence.
\end{IEEEproof}

\section{\texorpdfstring{Proof of \Cref{thm:general_partition_equilibrium}}{Proof of general partition equilibrium theorem}}
\label{app:thm:general_partition_equilibrium}

\begin{IEEEproof}
    By \Cref{prop:game_equivalence}, it is enough to work with the canonical
    utilities \(\widetilde u_{\mathrm S,n}=-(x_n-\theta_n)^2\) and
    \(\widetilde u_{\mathrm R,n}=-(x_n-g_n(\theta_n))^2\). Consider any finite
    PBE: If two messages induce the same
    action, they are payoff-equivalent for the Sender and can be merged without
    changing the outcome. Let the remaining distinct actions be ordered as
    \(x_{n,1}^*<\cdots<x_{n,\kappa_n}^*\).

    For the Sender, the preference between two adjacent actions is determined
    only by distance: a type \(\theta_n\) weakly prefers \(x_{n,i}^*\) to
    \(x_{n,i+1}^*\) iff
    \(\theta_n\le (x_{n,i}^*+x_{n,i+1}^*)/2\). Hence the Sender's best-response
    regions are contiguous intervals, with adjacent boundary
    \(\mu_{n,i}=(x_{n,i}^*+x_{n,i+1}^*)/2\).
    % Boundary types may be assigned to
    % either adjacent message because they have zero prior mass.

    Given a message \(m_{n,i}\), Bayes consistency identifies the
    posterior with the prior conditioned on \(I_{n,i}\). The Receiver therefore
    chooses \(x_n\) to maximize
    \(\mathbb{E}[-(x_n-g_n(\theta_n))^2\mid \theta_n\in I_{n,i}]\), whose unique
    maximizer is
    \(x_{n,i}^*=\mathbb{E}[g_n(\theta_n)\mid \theta_n\in I_{n,i}]\). This proves
    the necessity of
    \eqref{eq:best_response_theta}--\eqref{eq:arbitrage_cond_theta}.

    Conversely, suppose an ordered partition satisfies
    \eqref{eq:best_response_theta}--\eqref{eq:arbitrage_cond_theta}. Since
    \(g_n\) is strictly increasing and the prior has positive density,
    \eqref{eq:best_response_theta} implies
    \(x_{n,1}^*<\cdots<x_{n,\kappa_n}^*\). The midpoint conditions in
    \eqref{eq:arbitrage_cond_theta} then make \(x_{n,i}^*\) the closest action
    to every \(\theta_n\in I_{n,i}\), so sending \(m_{n,i}\) is optimal for all
    types in that interval. The Receiver is sequentially rational by
    \eqref{eq:best_response_theta}, and beliefs are Bayes-consistent. Thus the constructed assessment is a PBE.
\end{IEEEproof}

\section{\texorpdfstring{Proof of \Cref{prop:informative_delta}}{Proof of informative-delta proposition}}
\label{app:prop:informative_delta}
\begin{IEEEproof}
    By \Cref{lemma:existence_monotonicity}, an informative equilibrium exists if and only if a two-interval partition equilibrium exists. Such an equilibrium is defined by a boundary point $\mu_{n,1} \in (\underline{\theta}_n, \overline{\theta}_n)$ that solves the Sender's indifference condition from \eqref{eq:arbitrage_cond_theta}, which corresponds to a root of the continuous function $h(\mu) \coloneqq \frac{1}{2}(\mathbb{E}[g_n(\theta_n) | \theta_n \le \mu] + \mathbb{E}[g_n(\theta_n) | \theta_n > \mu]) - \mu$. By the Intermediate Value Theorem, a sufficient condition for a root to exist is that $h(\mu)$ has opposite signs at the boundaries of the type space $[\underline{\theta}_n, \overline{\theta}_n]$. The condition $\lim_{\mu \downarrow \underline{\theta}_n} h(\mu) > 0$ simplifies to $\frac{1}{2}(g_n(\underline{\theta}_n) + \mathbb{E}[g_n(\theta_n)]) - \underline{\theta}_n > 0$. Substituting $g_n(\theta_n) = \gamma_n\theta_n+\delta_n$ yields the lower bound on $2\delta_n$ in \eqref{eq:informative_exist_delta}. Similarly, the condition $\lim_{\mu \uparrow \overline{\theta}_n} h(\mu) < 0$ yields the upper bound, thus establishing the proposition.
\end{IEEEproof}

\section{\texorpdfstring{Proof of \Cref{prop:agreement_type}}{Proof of agreement-type proposition}}
\label{app:prop:agreement_type}
\begin{IEEEproof}
    The bias in canonical coordinates is
    \(\beta_n(\theta_n)=\theta_n-g_n(\theta_n)
    =(1-\gamma_n)\theta_n-\delta_n\). Since \(0<\gamma_n<1\),
    \(\beta_n\) is strictly increasing and has the unique zero
    \(\theta_n^\star=\delta_n/(1-\gamma_n)\)
    (the regularity condition required to apply \cite[Thms.~1, 2]{gordon2010on}).
    \emph{Case \shortref{prop:agreement_type:item:1}:} If \(\theta_n^\star\in
    [\underline{\theta}_n,\overline{\theta}_n]\), then
    \(\beta_n(\underline{\theta}_n)\le0\le
    \beta_n(\overline{\theta}_n)\). Thus the Sender's ideal action is weakly
    more extreme than the Receiver's at the two ends of the support, i.e., the
    subgame has outward bias in the sense of \cite{gordon2010on}, whose analysis covers state-dependent bias and applies directly to the affine bias $\beta_n$ here. Their
    infinite-equilibrium result gives an equilibrium with infinitely many
    messages, and the finite-partition implication gives an equilibrium with
    every finite number of messages. Hence
    \(\kappa_{n,\mathrm{max}}=\infty\).
    \emph{Case \shortref{prop:agreement_type:2}:} If \(\theta_n^\star\notin
    [\underline{\theta}_n,\overline{\theta}_n]\), then
    \(\beta_n\) is either strictly positive or strictly negative over the whole
    support. This is the strict one-sided bias case. The standard boundedness
    result for one-dimensional cheap-talk games with one-sided bias
    \cite{crawford1982strategic,gordon2010on} implies that only finitely many
    messages can be supported in equilibrium. Therefore
    \(\kappa_{n,\mathrm{max}}<\infty\).
\end{IEEEproof}

\section{\texorpdfstring{Proof of \Cref{prop:welfare_decomposition}}{Proof of welfare decomposition proposition}}
\label{app:prop_welfare_decomposition}
\begin{IEEEproof}
    Suppose the dispatch has a fixed active set \(\mathcal A\). The
    inactive consumers have \(x=0\), so their dispatch-dependent terms
    vanish.
    % and the constant \(-c\) cancels in welfare differences
    Let
    \(\bm Q_{\mathcal A}\coloneqq
    \mathrm{diag}((\alpha_n)_{n\in\mathcal A})+2a\bm 1\bm 1^\top\). The unconstrained first-order conditions \eqref{eq:individual_dispatch_marginal} can be written in matrix form as
    \(\bm x_{\mathcal A}^*(\bm v)
    =\bm Q_{\mathcal A}^{-1}(\bm v_{\mathcal A}-b\bm 1)\) for every admissible valuation \(\bm v\). This applies to both \(\bm v=\bm\omega\) and
    \(\bm v=\widehat{\bm\omega}\).

    Let \(\bm \epsilon\coloneqq
    \widehat{\bm\omega}_{\mathcal A}-\bm\omega_{\mathcal A}\). Because
    \(u_{\mathrm R}(\bm\omega,\cdot)\) is quadratic and its gradient vanishes at
    the full-information maximizer \(\bm x^*(\bm\omega)\), the second-order
    expansion is exact:
    \begin{equation}
        u_{\mathrm R}^{\mathrm{FC}}
        -
        u_{\mathrm R}\!\left(\bm\omega,\bm x^*(\widehat{\bm\omega})\right)
        =
        \frac{1}{2}\bm \epsilon^\top
        \bm Q_{\mathcal A}^{-1}\bm \epsilon.
        \label{eq:app_welfare_pointwise}
    \end{equation}
    Bayes consistency gives
    \(\epsilon_n=\mathbb{E}[\omega_n\mid m_n]-\omega_n\),
    \(\mathbb{E}[\epsilon_n\mid m_n]=0\), and
    \(\mathbb{E}[\epsilon_n^2]
    =\mathbb{E}[\mathrm{Var}(\omega_n\mid m_n)]\). Moreover, because the types
    are independent and each \(m_n\) is generated only from \(\omega_n\), the
    terms \(\{\epsilon_n\}_{n\in\mathcal A}\) are mutually independent; hence the
    cross terms vanish in expectation. Taking expectations in
    \eqref{eq:app_welfare_pointwise} gives
    \begin{equation}
        \mathbb{E}[u_{\mathrm R}^{\mathrm{FC}}]
        -
        \mathbb{E}[u_{\mathrm R}(\bm\sigma)]
        =
        \tfrac{1}{2}
        \sum_{n\in\mathcal A}
        [\bm Q_{\mathcal A}^{-1}]_{nn}\,
        \mathbb{E}\!\left[\mathrm{Var}(\omega_n\mid m_n)\right].
        \label{eq:app_welfare_diag}
    \end{equation}
    Finally, the coefficient of \(v_n\) in
    \([\bm Q_{\mathcal A}^{-1}(\bm v_{\mathcal A}-b\bm 1)]_n\) is
    \([\bm Q_{\mathcal A}^{-1}]_{nn}\). By \eqref{eq:exp-optimal-action-proof}, this coefficient is
    \(1/(\alpha_n+2a_n^{\mathrm{eff}})\). Substituting this identity into
    \eqref{eq:app_welfare_diag} yields \eqref{eq:welfare_decomposition}.
\end{IEEEproof}

\section{\texorpdfstring{Proof of \Cref{cor:consumer_payoff_comparison}}{Proof of consumer payoff comparison corollary}}
\label{app:cor_consumer_payoff_comparison}

\begin{IEEEproof}
    By \Cref{cor:equivalent_subgame}, the variance term due to other consumers
    cancels when comparing \(\bm{\sigma}\) and \(\bm{\sigma}'\). We therefore
    work in the effective subgame. Using the canonical type
    \(\theta_n \coloneqq (\omega_n - p)/\alpha_n\), the consumer's utility
    can be written as
    \(u_{\mathrm{S},n}(\omega_n, x_n) = -\frac{\alpha_n}{2}(x_n - \theta_n)^2 + \frac{\alpha_n}{2}\theta_n^2\),
    and the Receiver's best response is
    \(x_n = \gamma_n\mathbb{E}[\theta_n\mid m_n] + \delta_n\).

    By LIE,
    \(\mathrm{Cov}\bigl(\theta_n, \mathbb{E}[\theta_n\mid m_n]\bigr) = \mathrm{Var}\bigl(\mathbb{E}[\theta_n\mid m_n]\bigr)\).
    Expanding the consumer's quadratic loss yields
    \begin{align*}
        \mathbb{E}[(x_n-\theta_n)^2]
        = & \,\mathrm{Var}(\theta_n)-\gamma_n(2-\gamma_n)
        \mathrm{Var}\bigl(\mathbb{E}[\theta_n\mid m_n]\bigr)          \\
          & +\bigl((1-\gamma_n)\mathbb{E}[\theta_n]-\delta_n\bigr)^2.
    \end{align*}
    At a fixed tariff, only the second term depends on the signaling strategy (by message $m_n$).
    Transforming back to the type \(\omega_n\) with
    \(\mathbb{E}[\theta_n\mid m_n] =\tfrac{1}{\alpha_n}(\widehat{\omega}_n - p)\), we obtain
    \(\mathrm{Var}\bigl(\mathbb{E}[\theta_n\mid m_n]\bigr) = \mathrm{Var}(\widehat{\omega}_n)/\alpha_n^2\).
    % Subtracting the expected payoffs gives \eqref{eq:own_information_payoff_gain}, 
    % since \(\frac{\alpha_n}{2}\gamma_n \cdot \frac{1}{\alpha_n^2} = \frac{\gamma_n}{2\alpha_n} = c_n\).
    Evaluating the difference \(\mathbb{E}[u_{\mathrm{S},n}(\bm{\sigma}')] - \mathbb{E}[u_{\mathrm{S},n}(\bm{\sigma})]\)  yields \eqref{eq:own_information_payoff_gain},
    since \(\tfrac{1}{2}{\alpha_n}\gamma_n \cdot \frac{1}{\alpha_n^2} = \frac{\gamma_n}{2\alpha_n} = c_n\).
\end{IEEEproof}

\section{\texorpdfstring{Proof of \Cref{cor:pstar_alignment}}{Proof of price-alignment corollary}}
\label{app:cor_pstar_alignment}
\begin{IEEEproof}
    Fix $n \in \mathcal{A}$. By \Cref{cor:equivalent_subgame}, $(a_n^{\mathrm{eff}}, b_n^{\mathrm{eff}})$, and hence $\gamma_n \in (0,1)$, do not depend on $p$; recall $\theta_n^\star = (p - b_n^{\mathrm{eff}})/(2a_n^{\mathrm{eff}})$ and $\omega_n^\star(p) = \alpha_n \theta_n^\star + p$.
    \emph{Alignment:} substituting \eqref{eq:optimal_price_n} gives
    $p_n^* - b_n^{\mathrm{eff}} = \frac{2a_n^{\mathrm{eff}}(\mathbb{E}[\omega_n] - b_n^{\mathrm{eff}})}{\alpha_n + 2a_n^{\mathrm{eff}}}$, whence $\omega_n^\star(p_n^*) = \mathbb{E}[\omega_n]$ by direct computation, and $p_n^* = p^*$ for every $n$ by \Cref{prop:uniform_optimal_price}. Since $\psi_n$ is continuous and strictly positive on the nondegenerate interval $\Omega_n$, $\mathbb{E}[\omega_n] \in (\underline{\omega}_n, \overline{\omega}_n)$; the canonical restatement $\theta_n^\star = \mathbb{E}[\theta_n]$ follows from linearity of $\theta_n(\cdot)$ at $p = p^*$.
    \emph{Interior of the tariff bands:} $\mathrm{d}\omega_n^\star/\mathrm{d}p = 1 + \alpha_n/(2a_n^{\mathrm{eff}}) > 0$, so $\omega_n^\star(\cdot)$ is a continuous, strictly increasing bijection of $\mathbb{R}$; hence $\mathcal{P}_n \coloneqq \{p : \omega_n^\star(p) \in \Omega_n\}$ is the closed interval \eqref{eq:price_band_outward}, and since $\omega_n^\star(p^*) = \mathbb{E}[\omega_n]$ lies in the \emph{open} interval $(\underline{\omega}_n, \overline{\omega}_n)$, $p^*$ lies in the open interior of every $\mathcal{P}_n$, so $\bigcap_n \mathcal{P}_n$ contains a neighborhood of $p^*$. By the alignment property and~\Cref{prop:agreement_type:item:1}, $\kappa_{n,\mathrm{max}} = \infty$ for all $n$ at $p^*$.
    \emph{Uniqueness:} strict monotonicity of $\omega_n^\star(\cdot)$ and $p$-independence of $\mathbb{E}[\omega_n]$ give at most one solution of $\omega_n^\star(p) = \mathbb{E}[\omega_n]$; by the alignment property it is $p^*$.
\end{IEEEproof}

\section{\texorpdfstring{Proof of \Cref{prop:pbrd_convergence}}{Proof of P-BRD convergence proposition}}
\label{app:prop_pbrd_convergence}
\begin{IEEEproof}
    For \shortref{prop:pbrd_convergence:i}, we first establish the structural properties of the map $T_n$.
    For order preservation, by \eqref{eq:unprojected_T_def} the difference between adjacent updates for $i=1,\ldots,\kappa_n-2$ is
    \(
    [T_n(\bm{\mu}_n)]_{i+1} - [T_n(\bm{\mu}_n)]_i = \tfrac{1}{2}\bigl(x_{n,i+2}(\bm{\mu}_n) - x_{n,i}(\bm{\mu}_n)\bigr).
    \)
    Since $[\mu_{n,i-1},\mu_{n,i}]$ lies to the left of $[\mu_{n,i+1},\mu_{n,i+2}]$, their conditional means satisfy $\bar{m}_n(\mu_{n,i-1},\mu_{n,i}) \le \bar{m}_n(\mu_{n,i+1},\mu_{n,i+2})$. Because $g_n$ is strictly increasing, $x_{n,i}(\bm{\mu}_n) \le x_{n,i+2}(\bm{\mu}_n)$, which implies $[T_n(\bm{\mu}_n)]_{i+1} \ge [T_n(\bm{\mu}_n)]_i$. When $\bm{\mu}_n \in \mathcal{D}_n^\circ$, strict positivity of $\phi_n$ ensures that each $\bar{m}_n$ lies strictly in the interior of its cell, yielding strictly ordered actions and updated boundaries.

    Next, differentiating the conditional mean for $y<z$ gives
    \begin{align*}
        \frac{\partial\bar{m}_n(y,z)}{\partial y}
         & =\frac{\phi_n(y)(\bar{m}_n(y,z)-y)}
            {\int_y^z\phi_n(t)\,\mathrm{d}t}>0, \\
        \frac{\partial\bar{m}_n(y,z)}{\partial z}
         & =\frac{\phi_n(z)(z-\bar{m}_n(y,z))}
            {\int_y^z\phi_n(t)\,\mathrm{d}t}>0.
    \end{align*}

    % Log-concavity of $\phi_n$ further implies $\frac{\partial\bar{m}_n(y,z)}{\partial y} + \frac{\partial\bar{m}_n(y,z)}{\partial z} \le 1$ \cite[Lemma~3]{kazikli2022signaling,deimen2024communication}.
    Log-concavity of $\phi_n$ further implies $\frac{\partial\bar{m}_n(y,z)}{\partial y} + \frac{\partial\bar{m}_n(y,z)}{\partial z} \le 1$ \cite[Lemma~A.1]{deimen2024communication}, \cite[Lemma~3]{kazikli2022signaling}.
    Fix $y<z$, $c<d$, and parametrize the path $(y(t),z(t))=(1-t)(y,z)+t(c,d)$ for $t\in[0,1]$. Differentiating along the path gives
    \begin{align*}
        \left|\frac{\mathrm{d}}{\mathrm{d}t}
        \bar{m}_n(y(t),z(t))\right|
         & \le\frac{\partial\bar{m}_n}{\partial y}|c-y|
        +\frac{\partial\bar{m}_n}{\partial z}|d-z|      \\
         & \le\max\{|y-c|,|z-d|\}.
    \end{align*}
    Integrating establishes non-expansiveness:
    \begin{equation}
        |\bar{m}_n(y,z)-\bar{m}_n(c,d)|
        \le\max\{|y-c|,|z-d|\}.
        \label{eq:m_nonexpansive}
    \end{equation}
    By the convention $\bar{m}_n(y,y)=y$ and continuity of $\bar{m}_n$, \eqref{eq:m_nonexpansive} extends to all $\underline{\theta}_n\leq y\leq z\leq\overline{\theta}_n$.
    Applying \eqref{eq:m_nonexpansive} to \eqref{eq:brd_receiver_update}--\eqref{eq:unprojected_T_def} with $\mu_{n,0}=\nu_{n,0}=\underline{\theta}_n$ and $\mu_{n,\kappa_n}=\nu_{n,\kappa_n}=\overline{\theta}_n$ yields
    \(
    \|T_n(\bm{\mu}_n)-T_n(\bm{\nu}_n)\|_\infty \le \gamma_n\|\bm{\mu}_n-\bm{\nu}_n\|_\infty
    \)
    for all $\bm{\mu}_n,\bm{\nu}_n\in\mathcal{D}_n$.

    Combining order preservation of $T_n$ and monotonicity of $\Pi_n$ ensures that $\widetilde{T}_n(\mathcal{D}_n)\subseteq\mathcal{D}_n$. Because $\Pi_n$ is non-expansive,
    \begin{align*}
        \|\widetilde{T}_n(\bm{\mu}_n)-\widetilde{T}_n(\bm{\nu}_n)\|_\infty
         & \le\|T_n(\bm{\mu}_n)-T_n(\bm{\nu}_n)\|_\infty \\
         & \le\gamma_n\|\bm{\mu}_n-\bm{\nu}_n\|_\infty.
    \end{align*}
    The closed simplex $\mathcal{D}_n$ is complete under $\|\cdot\|_\infty$. Since $\gamma_n\in(0,1)$, the Banach fixed-point theorem~\cite{banach1922sur} guarantees a unique fixed point $\bm{\mu}_n^*\in\mathcal{D}_n$ and global geometric convergence
    \(
    \|\bm{\mu}_n^{(\ell)}-\bm{\mu}_n^*\|_\infty
    \le\gamma_n^\ell
    \|\bm{\mu}_n^{(0)}-\bm{\mu}_n^*\|_\infty.
    \)

    For~\shortref{prop:pbrd_convergence:ii}, if $\bm{\mu}_n^*\in\mathcal{D}_n^\circ$, projection is inactive at the fixed point, so $T_n(\bm{\mu}_n^*)=\bm{\mu}_n^*$. The induced actions are strictly ordered by positivity of $\phi_n$ and monotonicity of $g_n$. By \Cref{thm:general_partition_equilibrium}, \eqref{eq:best_response_theta}--\eqref{eq:arbitrage_cond_theta} define a strictly ordered $\kappa_n$-message PBE.
    Conversely, any strictly ordered $\kappa_n$-message PBE is an interior fixed point of $T_n$ and hence of $\widetilde{T}_n$. Uniqueness from~\shortref{prop:pbrd_convergence:i} implies it must coincide with $\bm{\mu}_n^*$. Thus, if $\bm{\mu}_n^*$ lies on the boundary of $\mathcal{D}_n$, no strictly ordered $\kappa_n$-message PBE exists.

    Finally, to obtain \eqref{eq:aposteriori_bound}, iterating the contraction and summing the geometric series yields
    \begin{align*}
        \|\bm{\mu}_n^{(\ell)}-\bm{\mu}_n^*\|_\infty
         & \le\sum_{j=0}^{\infty}
        \|\bm{\mu}_n^{(\ell+j+1)}-\bm{\mu}_n^{(\ell+j)}\|_\infty \\
         & \le\frac{\gamma_n}{1-\gamma_n}
        \|\bm{\mu}_n^{(\ell)}-\bm{\mu}_n^{(\ell-1)}\|_\infty
        =e_n^{(\ell)}.
    \end{align*}
    Hence, for each cell $i=1,\ldots,\kappa_n$,
    \(
    \mu_{n,i}^*-\mu_{n,i-1}^*
    \ge\mu_{n,i}^{(\ell)}-\mu_{n,i-1}^{(\ell)}
    -2e_n^{(\ell)}.
    \)
    Strict positivity of these lower bounds certifies that all equilibrium intervals are nondegenerate.
\end{IEEEproof}

\section*{Acknowledgment}
The authors would like to thank Benjamin Larrousse for his contribution to the initial phase of this work.

\bibliographystyle{IEEEtran}
\bibliography{references-ieee}

% ---------------------------------------------------------------------------
\ifshowbios
    \raggedbottom
    \authorbio{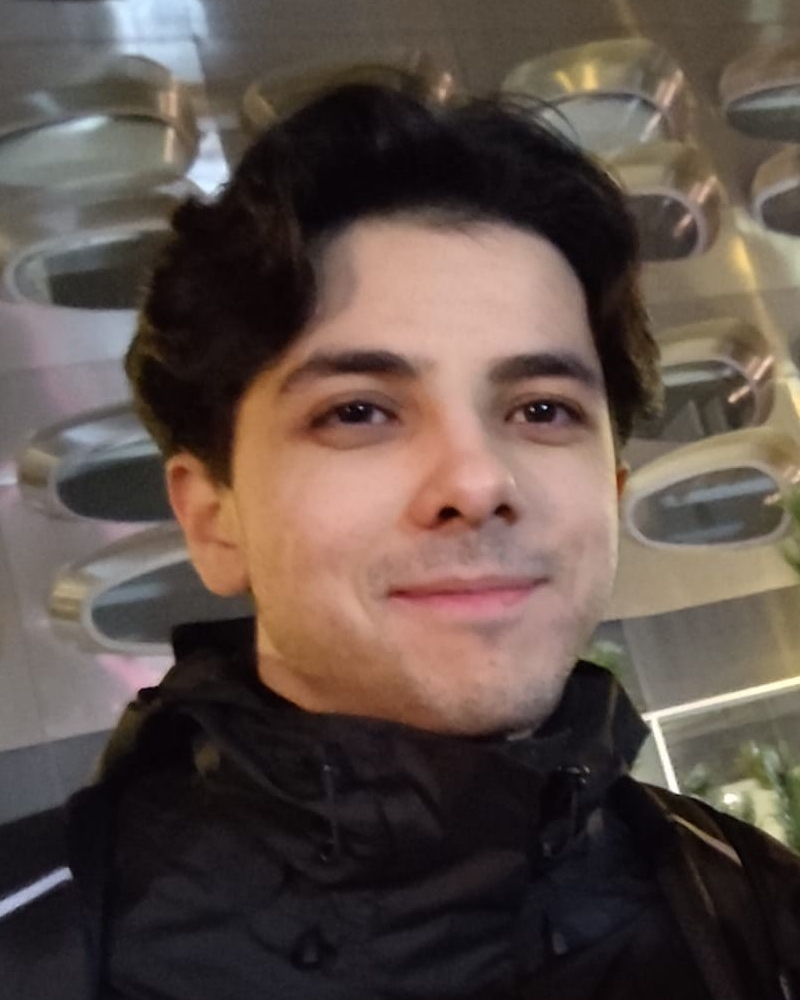}{Hassan Munif}{%
        received an Engineering Diploma in computer and telecommunication engineering from the Faculty of Engineering, Lebanese University, Beirut, Lebanon, in 2022, and the M.Sc. degree in multimedia networking from Université Paris-Saclay, Orsay, France, in 2023. He is currently pursuing the Ph.D. degree in automatic control with the CRAN Laboratory, Université de Lorraine, Nancy, France, and jointly with TICLab in the International University of Rabat.

        His research interests include game theory, opinion dynamics, multi-agent systems and federated learning.}

    \authorbio{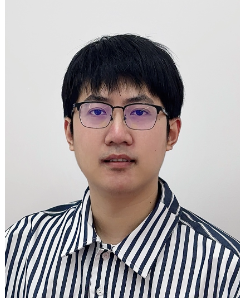}{Chao Zhang}{%
        received the B.S. degree in Optoelectronics from Huazhong University of Science and Technology, Wuhan, China, in 2012, and the M.S. and Ph.D. degrees in Electrical Engineering from the University of Paris-Saclay in 2014 and 2017, respectively. Prior to joining Central South University as an Assistant Professor in 2022, he held research positions at Centrale Supélec, Télécom Paris, and Princeton University.
        His research interests focus on large model compression, edge LLM, and semantic communication.}

    \authorbio{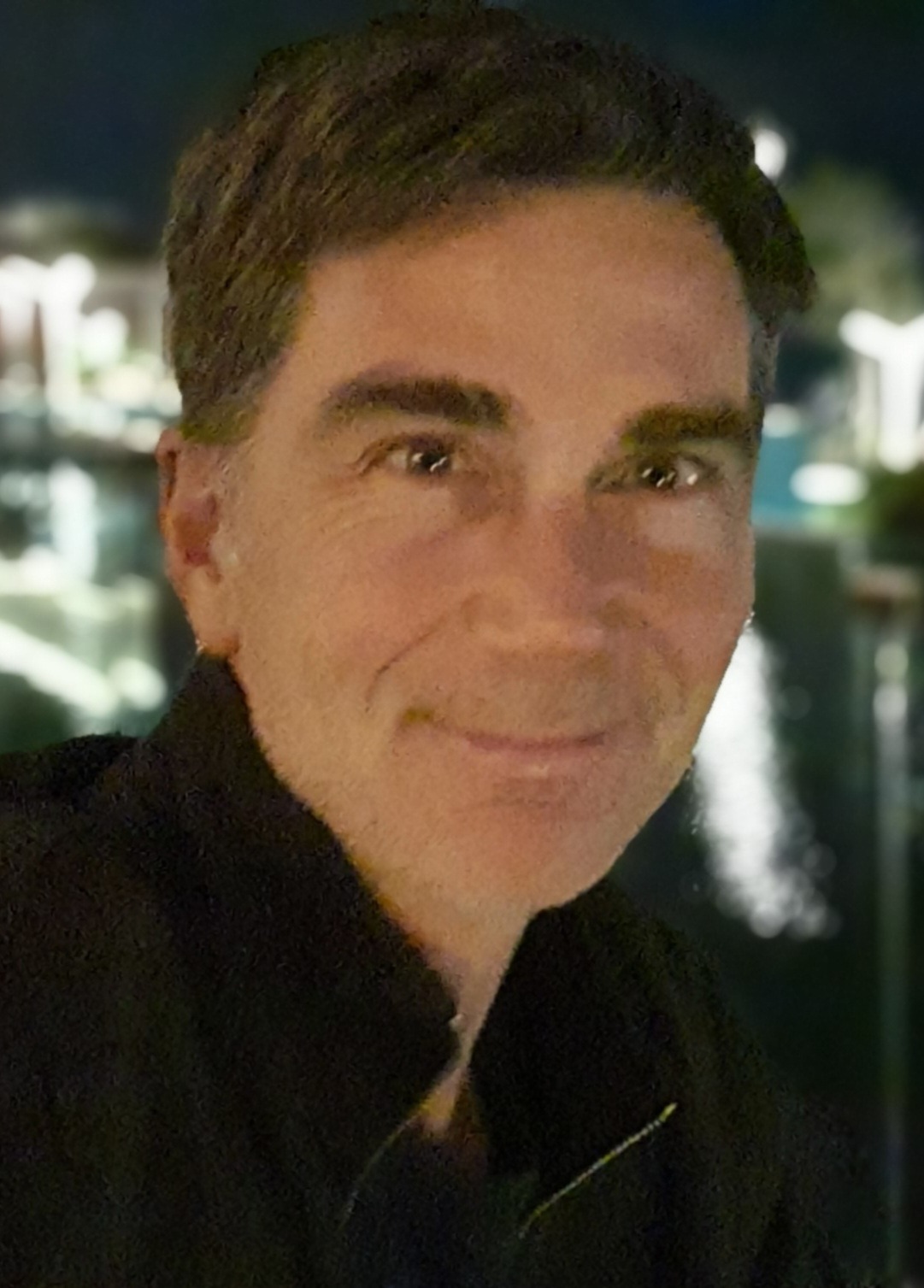}{Samson Lasaulce}{%
        (Member, IEEE) is a CNRS Director of Research with CRAN at Nancy. Over 2023-2025, he has been a Chief Research Scientist in AI with Khalifa University (KU), Abu Dhabi where he was also the holder of the TII 6G Chair on Native AI. He has been the holder of the RTE Chair on the Digital Transformation of Electricity Networks at CentraleSupélec. He has also been a Professor with the Department of Physics at École Polytechnique. Before joining CNRS he has been working for five years in private R\&D companies (Motorola Labs and Orange Labs). Dr. Lasaulce is the recipient of several awards such as the SEE Blondel Medal. Dr. Lasaulce has been serving as an Associate Editor for several journals such as the IEEE Transactions on Signal Processing and Springer Nature Discover Artificial Intelligence. His current research interests lie in distributed networks with a focus on game theory, information theory, learning, distributed optimization, network control for communication networks, energy networks, and social networks. He is a co-author of the books ``Game Theory and Learning for Wireless Networks: Fundamentals and Applications'' and ``Game Theory and Learning for Engineers: Fundamentals and Applications''.}

    \authorbio{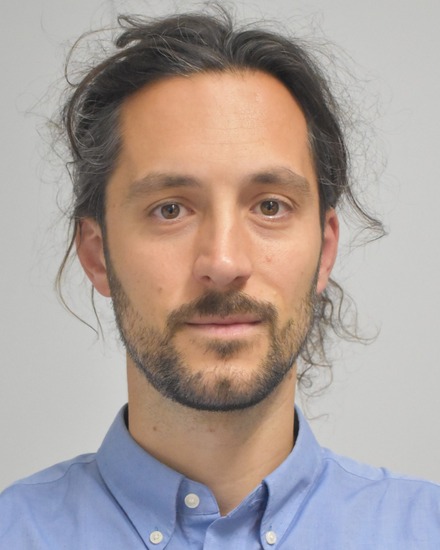}{Olivier Beaude}{%
        received the M.Sc. degree in applied mathematics and economics from École Polytechnique, Palaiseau, France, in 2010, the M.Sc. degree in optimization, game theory, and economic modeling from Université Paris VI, Paris, France, in 2011, and the Ph.D. degree from the Laboratory of Signals and Systems, CentraleSupélec, and the Research and Development Center, Renault, France, in 2015. He is currently a Research Engineer with EDF Research and Development, where his research interests include controlling the impact of EV charging on the grid, the development of coordination mechanisms, and incentives for local electricity systems, and long-term studies (Generation Expansion Planning, decarbonization) -- with a specific focus on island(ed) systems.}

    \newpage
    \vspace*{56.7pt} 
    \authorbio{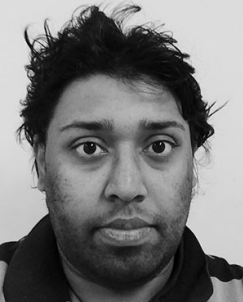}{Vineeth S. Varma}{%
        received the bachelor's in physics degree with Honors from Chennai Mathematical Institute, Chennai, India, in 2008, dual master's in science and technology degree from the Friedrich-Schiller-University of Jena, Germany, in 2009, and Warsaw University of Technology, Warsaw, Poland, in 2010, and the Ph.D. degree in physics from LSS-University of Paris Saclay, SaintAubin, France. He is currently a CNRS researcher at the Centre de Recherche en Automatique de Nancy (CRAN) in Nancy, France.

        His areas of interest are energy efficiency in telecommunication, multi-agent systems, and the interface of automatic control and communication.}

    \authorbio{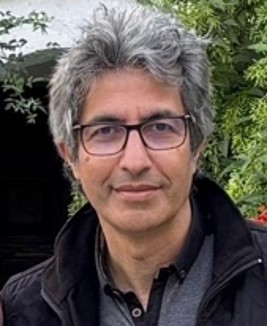}{Mounir Ghogho}{%
        (F'18) received the Ph.D. degree in signal processing from the National Polytechnic Institute of Toulouse, France, in 1997. From 1997 to 2001, he was a Research Fellow with the University of Strathclyde, Scotland. He then joined the University of Leeds, England, where he was appointed Full Professor and Chair in Signal Processing and Communications in 2008. In 2010, he joined the International University of Rabat, where he founded TICLab and held several academic leadership roles. Since March 2025, he has been with the College of Computing at Mohammed VI Polytechnic University. His research interests include machine learning and statistical signal processing, with applications in wireless communications, robotics, cybersecurity, and healthcare. He is an IEEE Fellow and a Fellow of the International Academy of Artificial Intelligence Sciences and The World Academy of Sciences. He has supervised over 50 Ph.D. students and has held editorial, technical-committee, and conference-leadership roles within IEEE.}

    \authorbio{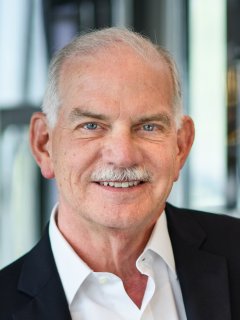}{H. Vincent Poor}{%
        (S'72, M'77, SM'82, F'87) received the Ph.D. degree in EECS from Princeton University in 1977.  From 1977 until 1990, he was on the faculty of the University of Illinois at Urbana-Champaign. Since 1990 he has been on the faculty at Princeton, where he is currently the Michael Henry Strater University Professor. During 2006 to 2016, he served as the dean of Princeton's School of Engineering and Applied Science, and he has also held visiting appointments at several other universities, including most recently at Berkeley and Caltech. His research interests are in the areas of information theory, stochastic analysis and machine learning, and their applications in wireless networks, energy systems and related fields. Among his publications in these areas is the book Advanced Data Analytics for Power Systems  (Cambridge University Press, 2021). Dr. Poor is a member of the National Academy of Engineering and the National Academy of Sciences and is a foreign member of the Royal Society and other national and international academies.}
\fi
% ---------------------------------------------------------------------------

\end{document}